%% file: main.tex
\documentclass[11pt]{article}

\usepackage[usenames,dvipsnames]{xcolor}
\usepackage[pagebackref,letterpaper=true,colorlinks=true,pdfpagemode=none,urlcolor=blue,linkcolor=blue,citecolor=BrickRed,pdfstartview=FitH]{hyperref}
\usepackage{prettyref}
\usepackage{amsmath}
\usepackage{amsthm}
\usepackage{amssymb}
\usepackage{algorithm}
\usepackage{color}
\usepackage[english]{babel}

\usepackage{graphicx}
\usepackage{caption}
\usepackage{subcaption}
\newcommand{\Mon}{\mathsf{Mon}}
\usepackage[square,sort,comma,numbers]{natbib}
\usepackage{wrapfig,epsfig}
\usepackage{psfrag}
\usepackage{epstopdf}
\usepackage{url}
\usepackage{graphicx}
\usepackage{color}
\usepackage{epstopdf}
\usepackage{algpseudocode}
\usepackage{hyperref}
\usepackage[margin=1in]{geometry}
\newtheorem{theorem}{Theorem}[section]
\newtheorem{lemma}[theorem]{Lemma}
\newtheorem{definition}[theorem]{Definition}

\newtheorem{corollary}[theorem]{Corollary}

\newtheorem{fact}[theorem]{Fact}
\newtheorem{remark}[theorem]{Remark}
\newtheorem{claim}[theorem]{Claim}

\newcommand{\bi}{\boldsymbol{i}}
\newcommand{\bx}{\boldsymbol{x}}
\newcommand{\bX}{\boldsymbol{X}}
\newcommand{\bZ}{\boldsymbol{Z}}

\newcommand{\gap}{\mathsf{gap}}
\newcommand{\bg}{\boldsymbol{g}}

\newcommand{\wh}{\widehat}
\newcommand{\wt}{\widetilde}

\def\bx{{\bf x}}
\newcommand{\ignore}[1]{}

\newcommand{\eps}{\varepsilon}

\newcommand{\supp}{\mathsf{supp}}
\newcommand{\poly}{\mathsf{poly}}

\DeclareMathOperator*{\argmin}{arg\,min}
\DeclareMathOperator*{\E}{\mathbb{E}}

\newcommand{\FF}{\mathcal{F}}
\newcommand{\fF}{\mathbf{F}}

\newcommand{\midd}{\mathsf{Apx}}
\newcommand{\sign}{\mathsf{sign}}
\newcommand{\nulll}{\mathsf{null}}

\newenvironment{proofof}[1]{\bigskip \noindent {\it Proof of #1.}\quad }
{\qed\par\vskip 4mm\par}

\input{libtheorems.tex}

\begin{document}

\title{Reconstruction under outliers for Fourier-sparse functions}
\author{
   Xue Chen\\\
  \texttt{xue.chen1@northwestern.edu}\\
  Northwestern University
 \and   Anindya De\thanks{Supported by NSF grant CCF 1926872 (transferred from CCF-1814706)}\\
  \texttt{anindyad@seas.upenn.edu}\\
  University of Pennsylvania
}

\begin{titlepage}

\maketitle

\begin{abstract}
We consider the problem of learning an unknown $f$  with a sparse Fourier spectrum in the presence of outlier noise. In particular, the 
algorithm has access to a noisy oracle for (an unknown) $f$ such that (i) the Fourier spectrum of $f$ is $k$-sparse; (ii) at any query point $x$, the oracle 
returns $y$ such that with probability $1-\rho$, 
 $|y-f(x)| \le \epsilon$.  However, with probability $\rho$, the error $y-f(x)$ can be arbitrarily large.

 We study Fourier sparse functions over both the discrete cube $\{0,1\}^n$ and the torus $[0,1)$ and for both these domains, we design efficient algorithms which can tolerate any $\rho<1/2$ fraction of outliers. We note that the analogous problem for low-degree polynomials has recently been studied in several works~\cite{AK03, GZ16, KKP17} and similar algorithmic guarantees are known in that setting. 
 
 While our main results pertain to the case where the location of the outliers, i.e., $x$ such that $|y-f(x)|>\epsilon$ is randomly distributed, we also study the case where the outliers are adversarially located. In particular, we show that over the torus, assuming that the Fourier transform satisfies a certain \emph{granularity} condition, there is a sample efficient algorithm to tolerate $\rho =\Omega(1)$ fraction of outliers and further, that this is not possible without such a granularity condition. Finally, while not the principal thrust, our
 techniques also allow us non-trivially improve on learning  low-degree functions $f$ on the hypercube in the presence of adversarial outlier noise. 
 
 Our techniques combine a diverse array of tools from compressive sensing, sparse Fourier transform, chaining arguments and complex analysis.

 \end{abstract}
\thispagestyle{empty}

\end{titlepage}

\input{intro}

\input{proof_sketchv1}

\input{preli}

\input{LP_decoding-ad}

\input{BooleanCube}

\input{torus}

\input{adv_outliers_boolean}

\input{adv_outliers_torus}

\bibliographystyle{alpha}
\bibliography{wildnoise}


\appendix
\newpage
\input{appen_proof}

\end{document}

%% file: libtheorems.tex
\usepackage{etoolbox}

\makeatletter

\newcommand{\define}[4][ignore]{%
  \ifstrequal{#1}{ignore}{}{
  \@namedef{thmtitle@#2}{#1}}%
  \@namedef{thm@#2}{#4}%
  \@namedef{thmtypen@#2}{lemma}%
  \newtheorem{thmtype@#2}[theorem]{#3}%
  \newtheorem*{thmtypealt@#2}{#3~\ref{#2}}%
}

\newcommand{\state}[1]{%
  \@namedef{curthm}{#1}
  \@ifundefined{thmtitle@#1}{
  \begin{thmtype@#1}
    }{
  \begin{thmtype@#1}[\@nameuse{thmtitle@#1}]
  }
    \label{#1}
    \@nameuse{thm@#1}
  \end{thmtype@#1}
  \@ifundefined{thmdone@#1}{
  \@namedef{thmdone@#1}{stated}%
  }{}
}

\newcommand{\restate}[1]{%
  \@namedef{curthm}{#1}
  \@ifundefined{thmtitle@#1}{
    \begin{thmtypealt@#1}
    }{
  \begin{thmtypealt@#1}[\@nameuse{thmtitle@#1}]
  }
    \@nameuse{thm@#1}
  \end{thmtypealt@#1}
  \@ifundefined{thmdone@#1}{
  \@namedef{thmdone@#1}{stated}%
  }{}
}

\newcommand{\thmlabel}[1]{
  \@ifundefined{thmdone@\@nameuse{curthm}}{\label{#1}
    }{\tag*{\eqref{#1}}}
}
\makeatother

%% file: intro.tex
\section{Introduction}~\label{sec:intro}
The starting point of this paper is the following fundamental  algorithmic problem --   there is an unknown signal (or function) $f: \mathcal{D} \rightarrow \mathbb{C}$ (on some known domain $\mathcal{D}$). The algorithm can query $f(x)$ for any $x \in \mathcal{D}$ and the goal is to  recover $f$ with query complexity much smaller than $|\mathcal{D}|$. Can \emph{structural knowledge} about $f$ permit such efficient recovery algorithms? To motivate this, we consider two such instances of \emph{structural knowledge}. 
 Throughout this paper, our domain $\mathcal{D}$ is one of the following: (i) the $n$-dimensional hypercube $\{0,1\}^n$ or (ii) the one-dimensional torus $\mathbb{R}/\mathbb{Z}$ which is equivalently $[0,1)$. 
~\\
 \textbf{Low-degree polynomials:} Let us assume that $f$ is a degree-$d$ polynomial and $\mathcal{D}$ is either $[0,1)$ or $\{0,1\}^n$. Recovering $f$ is then the same as polynomial interpolation and can be done by making $N_d$ queries and solving a linear system where $N_d = d+1$ for the torus and $N_d = \binom{n}{\le d}$ for the discrete cube (Recall that $\binom{n}{\le d}$ is defined to be $\sum_{0 \le  j \le d} \binom{n}{j}$). 
 ~\\
 \textbf{Fourier sparse signals:} The second kind of structural assumption that has been extensively studied
 in literature is \emph{Fourier sparsity}. Namely, assume that the Fourier transform of $f$ (denoted by $\widehat{f}$) is $k$-sparse, i.e., non-zero in at most $k$ positions. Compared to the case when $f$ is low-degree, this case turns out to be significantly more difficult. When $\mathcal{D} = 
\{0,1\}^n$, 
 the seminal work of Goldreich and Levin~\cite{GoldreichLevin:89} gives an efficient recovery algorithm with $\mathsf{poly}(k, n)$ query complexity. For the torus, the earliest work on this problem dates back to Prony (from 1795). In theoretical computer science,
 this problem was first studied  by Gilbert \emph{et al.}~\cite{gilbert2002near} who achieved a $\mathsf{poly}(k,\log F)$ query and time complexity for this problem  -- here $F$ is the bandlimit, i.e., all the non-zero frequencies of $f$ are assumed to lie in $[-F, \ldots, F]$. 
  In fact, this work was also the starting point of a rich line of work on sparse FFTs \cite{gilbert2002near,AGS03,GMS05,HIKP12b,IK14,Kapralov16,CKPS17,AKMMVZ18} where by now nearly optimal query and time complexity are known. 
 
  \subsection*{Robust recovery problems} 
  So far, the problem statement assumes that the oracle to $f: \mathcal{D} \rightarrow \mathbb{C}$ is noise free. However, from the vantage point of both intellectual interest and practical applications, it is important to consider the case where the oracle to $f$ is noisy. In particular, we are interested in a strong model of noise called \emph{outlier noise}. 
  \begin{definition} ~
  \newline
  {\textbf{Random outlier noise:}} An oracle  for $f: \mathcal{D} \rightarrow \mathbb{C}$ is said to have $(\rho, \epsilon)$-\emph{random outlier noise}   if on any input $x \in \mathcal{D}$, the oracle returns $y(x)$ such that 
  $|y(x) -f(x)| \le \epsilon$ with probability $1-\rho$. ~
  \newline
  {\textbf{Adversarial outlier noise:}}
  An oracle for $f: \mathcal{D} \rightarrow \mathbb{C}$  is said to have $(\rho, \epsilon)$-adversarial outlier noise  if $\Pr_{x \sim \mathcal{D}}[|y(x) -f(x)| \le \epsilon] \ge 1-\rho$. We recall that for a set $\mathcal{D}$, $x \sim \mathcal{D}$ means choosing an element of $\mathcal{D}$ uniformly at random. 
  \end{definition}
Let us call $x$ to be an \emph{outlier} (in either of the models above) if $|y(x) - f(x)|>\eps$  -- otherwise, we call the point an \emph{inlier}. 
  We clarify that in both the models above, the target function $f$ is fixed and unknown to the learner (and is not randomized in any sense). In the random outlier noise model, 
it is only the location of the outliers which are random. The value of the error, i.e., $y(x) - f(x)$ is adversarial, for both inliers and outliers.
 We now turn to a brief discussion of some prior work. 
  
  \textbf{Robust recovery for low-degree polynomials:} 
Arora and Khot~\cite{AK03} were the first to study the problem of robust recovery of polynomials under outlier noise. They worked in the \emph{random outlier noise model} with the domain $\mathcal{D}=[0,1)$. 
Here, they showed that even if we forego computational efficiency, $\rho<1/2$ is required for any non-trivial recovery ({even in the sense of list decoding}) and achieved a computationally efficient  algorithm when $\rho=0$. 
This was significantly improved in a recent work by Guruswami and Zuckerman~\cite{GZ16} who achieved computationally efficient algorithm for all $\rho \leq 1/\log d$ ($d$ is the degree of the polynomial). Finally, Kane, Karmalkar and Price~\cite{KKP17} improved this to obtain computationally efficient algorithms for any $\rho <1/2$. In the adversarial setting the information theoretic upper bound on $\rho$ is $O(1/d^2)$ and there are computationally efficient algorithms achieving this bound (see \cite{GZ16}).

\textbf{Robust recovery for Fourier sparse signals:} 
While the recovery of Fourier sparse signals has attracted much attention (in the context of sparse FFTs), robustness considerations have mainly been restricted to noise bounded in a norm such as $\ell_1$ \cite{Iwen2010,BCGLS,MI17} or $\ell_2$ \cite{gilbert2002near,GMS05,HIKP12b,IK14,Kapralov16}. However, to the best of our knowledge, the problem of recovery of such signals in the outlier model, i.e., the noise is not bounded in any $\ell_p$ norm has not been considered in the literature previously. 

The principal conceptual challenge in obtaining robust recovery results for Fourier sparse signals vis-a-vis low degree polynomials is that a degree-$k$ polynomial admits a sparse representation in terms of \emph{known basis elements}, i.e., monomials of degree at most $k$. In contrast, a $k$-Fourier sparse signal admits a sparse representation in terms of $k$ unknown Fourier characters. Note that whole set of Fourier characters has size $|\mathcal{D}| \gg k$. One of course can resort to an 
exhaustive enumeration over all $k$ subsets of Fourier characters -- however, such an enumeration is computationally prohibitive.


\paragraph{Our results}
The main results of this paper are efficient algorithms that recover  Fourier sparse signals under  random outliers. {In the results below, for quantities $x$ and $y$, when we say $x \lesssim y$, we mean that there is a constant $c>0$ such that $x \le c \cdot y$. For a function $f$ with the Fourier transform $\wh{f}$, let $\supp(\wh{f})$ denote the support of $\wh{f}$, i.e., the subset of its Fourier characters with non-zero coefficients.}

We first state our result for the Boolean cube.  

\begin{theorem}\label{thm:inf_boolean_rand_out}[Informal version of Theorem \ref{thm:main_random_boolean_cube}]
There exists an algorithm which given as input, sparsity parameter $k$ and additional parameters $\eta>0$, $0 \le \rho <1/2$ and input domain $\{0,1\}^n$, makes $\tilde{O}(k^2 n)$  queries and runs in time $\poly(k,n,\frac{1}{\eta})$ such 
 that for any $f(x)=\sum_{i=1}^k \wh{f}(\xi_i) (-1)^{\langle \xi_i, x \rangle}$ with each $|\wh{f}(\xi_i)| \ge \eta$, under  $(\rho, \eps )$ random outlier noise (with $\eps \lesssim \eta$), its output $g$ satisfies
$$
\supp(\wh{g})=\supp(\wh{f}) \text{ and } \E_x[|g(x)-f(x)|] \lesssim \eps, \textit{ with probability } 0.99. $$
In particular, this implies $|\wh{g}(\xi)-\wh{f}(\xi)| \lesssim \eps$ for every $\xi$.
\end{theorem}
We note that the constraints, $\rho < 1/2$ and $\eps \lesssim \eta$ are necessary.  To see this, (i) note that even for the  subcase of low-degree  polynomials, $\rho$ needs to be less than $1/2$ (observed by \cite{AK03}). (ii) Similarly, consider a $1$-Fourier sparse function $f$ such that $|\wh{f}(\xi)| = \eps$ if $\xi = \xi_0$ (for some specific $\xi_0$) and $0$ otherwise. Under $(0, \epsilon)$-outlier noise, $f$ is the same as the function $g(x)$ which is $0$ everywhere, thus making it impossible to distinguish $f$ and $g$. Consequently, we require $\eps \lesssim \eta$.


The next theorem is an analogous result for the torus $[0,1)$. Unlike the domain $\{0,1\}^n$, the torus is infinite and thus has infinitely many Fourier characters. So, it  becomes necessary to assume that all the frequencies appearing in the spectrum of $f$ come from some interval $[-F, F]$ (referred to as the bandlimit of $f$). 
\begin{theorem}\label{thm:inf_periodic_rand_out}[Informal version of Theorem \ref{thm:periodic_FFT_rand_out}]
There exists an algorithm which given as input, sparsity parameter $k$, bandlimit $F$ and additional parameters $\eta>0$, $0 \le \rho <1/2$ and input domain $[0,1)$, 
makes $\tilde{O}(k^2 \log F)$ queries and runs in $\poly(k, \log F, 1/\eta)$ time such that for any $f(t)=\sum_{j=1}^k \wh{f}(\xi_j) \cdot e^{2 \pi \bi \xi_j \cdot t}$ with each $\xi_j \in [-F,F]$ and $|\wh{f}(\xi_j)| \ge \eta$, under the $(\rho, \eps )$ random outlier noise of $\eps \lesssim \eta$, its output $g$ satisfies 
$$
\supp(\wh{g})=\supp(\wh{f}) \text{ and } \E_x[|g(x)-f(x)|] \lesssim \eps, \textit{ with probability } 0.99.$$
In particular, this implies $|\wh{g}(\xi)-\wh{f}(\xi)| \lesssim \eps$ for every $\xi$. Further, with high probability all the query points of the algorithm are $1/\poly(k, \log F)$ apart from each other. 
\end{theorem}

{{ 
We now briefly discuss the importance of the last condition -- namely, any two query points are separated by $1/\poly(k, \log F)$. Such a separation, (as opposed to allowing the query points to be say $1/F$ close to each other)  is crucial for the result to be non-trivial. To see this, consider the following  obvious \emph{outlier-removal routine}.
For any point $x$, define $\mathcal{K}_x = [x- \vartheta, x+ \vartheta]$ where $\vartheta>0$ is any sufficiently small quantity. Note that by taking $\vartheta$ sufficiently small, $|f(x) - f(z)| < \epsilon$ for all $z \in \mathcal{K}_x$.  In fact, in our setting,  taking $\vartheta \le \epsilon/F$, suffices to ensure this.

 Now suppose the algorithm has access to an oracle with $(\rho, \epsilon)$-random outlier noise.  Let $\rho =1/2-\delta$. 
 Then, instead of querying the oracle at $x$ (to obtain $f(x)$), 
  the algorithm queries the oracle at $O(1/\delta^2)$ randomly chosen points in $\mathcal{K}_x$ and outputs the median.  Using the fact that the median is tolerant to presence of outliers, it is easy to see that with high probability, this procedure returns a value $v$ such that $|v-f(x)| < \epsilon$, thus removing outlier noise. At this point, one can use an algorithm tolerant to bounded noise (such as  \cite{gilbert2002near}) 
 to  recover $f$.

 When the query points are required to be $1/\poly(k, \log F)$ apart (as is the case in Theorem~\ref{thm:inf_periodic_rand_out}), this simple procedure no longer works. Finally, we also remark  that \cite{GZ16}, where the goal was to tolerate outlier noise for low-degree (trigonometric) polynomials over the torus, 
 also had a similar requirement on the query points being ``sufficiently far" in order  to ensure a non-trivial problem.  
 

} Finally, we note that all previous algorithms such as \cite{gilbert2002near,AGS03,GMS05,HIKP12b,IK14,Kapralov16}  which compute sparse Fourier transforms are unable to tolerate more than $\rho=\Omega(1/k)$ fraction of outliers. What allows Theorem~\ref{thm:inf_boolean_rand_out} and Theorem~\ref{thm:inf_periodic_rand_out} to improve on this is a combination of two ideas: 
\begin{enumerate}
\item[(a)] \emph{Domain reduction} -- An idea coming from the sparse Fourier transform literature which allows us reduce  the underlying domain to size comparable to the sparsity of $f$. 
\item[(b)] \emph{Linear programming} -- With a reduced domain size we employ a linear program which allows us to recover the underlying signal while tolerating up to $\rho<1/2$ fraction of random outliers. The running time of the linear program is polynomial in the domain size. Note that linear programming has also been used as an algorithmic tool in other \emph{sparse recovery type problems} such as compressive sensing~\cite{CRT06} and LP decoding~\cite{DMT07}. 
\end{enumerate}
}

\paragraph{Adversarial outlier noise} We now turn our discussion to the adversarial outlier noise model. 
The  fundamental bottleneck here is  that for a $k$-Fourier-sparse function $f$, its entire $\ell_2$ mass could essentially concentrate on a $1/k$ fraction of the domain. Note that this is tight by the so-called ``uncertainty principle". However, this means that with $\rho \ge 1/k$ fraction of adversarial outliers, a noisy oracle can return $0$ on the entire set $\{x: f(x) \not  =0\}$, thus making it indistinguishable from the function which is identically zero on the domain. 
We now discuss two conditions under which one can circumvent the above barrier at $\rho= 1/k$. ~\\
\textbf{Low-degree polynomials:} As we highlighted earlier, the principal challenge in recovering a Fourier sparse signal $f$ is that while it admits a sparse representation in the Fourier basis, the basis elements with non-zero coefficients are not known. However, in many cases of interest, say constant depth circuits~\cite{linmannis93} from circuit complexity, the Fourier spectrum is not only (approximately) sparse but also low-degree. Motivated by this, we consider the problem of recovering $f: \{0,1\}^n \rightarrow \mathbb{R}$ where $f$ is a degree-$d$ polynomial (i.e., supported on Fourier characters of size, i.e., Hamming weight, at most $d$). Observe that there are $\binom{n}{\le d} \approx n^d$ such Fourier characters -- thus, if the ``uncertainty principle" were tight, then we could only tolerate $\approx n^{-d}$ fraction of adversarial outliers.  However, we show that for degree-$d$ polynomials $f$, one can tolerate up to {$\Theta(3^{-2d})$} fraction of adversarial outliers. This result relies on so-called ``hypercontractivity of low-degree polynomials"~\cite{ODBook}. We refer the reader to Theorem~\ref{thm:degree_d_adv_outliers} for the precise theorem statement. 
~\\
\textbf{Granular coefficients:}  Another restricted case of Fourier sparse signals that we consider are \emph{granular signals}. Namely, let $f$ be a Fourier sparse  function over the torus such that $\wh{f}$ is $k$-sparse and further if $\wh{f}(\xi)$ is non-zero, then both its real and imaginary parts are integral multiples of some given number $\eta$. In this case, we 
give a sample efficient (though not computationally efficient) algorithm which can recover $f$ and can tolerate $(\rho, \epsilon)$ adversarial outlier noise with $\rho$ approaching $1/2$. This result relies on a certain \emph{anti-concentration property} of harmonic functions (from complex analysis). See Theorem~\ref{thm:periodic_adversarial_outlier} for a precise theorem statement along with tradeoffs between $\eta$, $\rho$ and $\epsilon$ that our algorithm can achieve. 
We note that the assumption on granularity of amplitudes has been used in literature before such as in the celebrated work of Hassanieh \emph{et~al.}~\cite{HIKP12b}. 


\paragraph{Our techniques.} 
Our results are obtained by drawing on a rich set of tools from compressed sensing \cite{CRT06}, sparse Fourier transform \cite{gilbert2002near,AGS03,GMS05,HIKP12b}, chaining arguments \cite{RV08}, and anti-concentration by hypercontractivity and complex analysis~\cite{ODBook, borwein1997zeros}. As mentioned earlier, our algorithm is a combination of the classical $\ell_1$ regression used in compressed sensing along with dimension reduction arguments coming from the  sparse Fourier transform literature. The analysis of the linear program crucially relies on a certain $\ell_1$ concentration property (see Claim~\ref{clm:ell_1_ball}) --  in turn, the proof of this $\ell_1$ concentration property relies on a chaining argument with tools coming from high dimensional probability theory.

\paragraph{Related work.} 
{As we said earlier, previous literature on recovering Fourier sparse functions assumes that the noise is bounded in some $\ell_p$ norm. 
In particular, 
most sparse FFT algorithms \cite{gilbert2002near,GMS05,HIKP12b,IK14,Kapralov16}  are randomized and provide $\ell_2/\ell_2$ guarantee, namely the output $g$ satisfies $\|g-f\|_2 \le C \cdot \|y-f\|_2$ when the noise is $\ell_2$ bounded. A small but intruiging strand of work~\cite{Iwen2010,MI17} also considers the design of \emph{deterministic} sparse FFT algorithms with $\ell_1/\ell_1$ guarantee --- $\|g-f\|_1 \le C \cdot \|y-f\|_1$. 
}

{Probably the line of work which most directly inspires ours is the one on 
fitting polynomials with outliers. This problem has been studied extensively in computer vision and machine learning (see \cite{AK03,GZ16} and the references therein). As mentioned earlier, Arora and Khot \cite{AK03}, Guruswami and Zuckerman \cite{GZ16}, and Kane, Karmalkar, and Price \cite{KKP17} have studied the reconstruction of \emph{low degree} polynomials (including low degree trigonometric polynomials in \cite{GZ16}) under random outliers. In particular, note that saying $f$ is a degree-$d$ trigonometric polynomial is equivalent to $\widehat{f}$ being supported on the first $d$ Fourier characters of the torus. We emphasize that our results are incomparable with this line of work~\cite{AK03,GZ16,KKP17}. On one hand, our setting is more challenging because unlike low-degree trigonometric polynomials, we do not explicitly know the (sparse) support of the Fourier spectrum. On the other hand, \cite{GZ16, KKP17}  recover a function $g$ with a guarantee $\|g-f\|_{\infty}=O(\eps)$. In contrast, we  guarantee closeness of $f$ and $g$ in the Fourier space, i.e., $\|\wh{g} - \wh{f} \|_{\infty} = O(\epsilon)$.} 
%
%
%
%
%

{A second strand of related work comes  from compressed sensing where a line of research has  focused on reconstruction from Gaussian and subgaussian measurements (such as 
linear measurements where each coefficient is an independent $\pm 1$ random variable) 
 with outliers \cite{DMT07,LDB09,NTN11,FM14,KP19}. Technically, Gaussian and $\{ \pm 1\}$ measurements provide much stronger concentration and anti-concentration than Fourier measurements -- this makes it possible to tolerate $\rho=\Theta(1)$ fraction of adversarial outliers. As mentioned earlier, with Fourier measurements, we provably cannot recover under such a strong model of noise. 
  One exception here is the work of  Nguyen and Tran \cite{NT13} who studied compressed sensing with random outliers using Fourier measurements.  However, their model is weaker than ours in several ways: firstly the support of the (sparse) Fourier spectrum is randomly distributed and each entry in the support is equally likely to be either $1$ or $-1$; secondly when the bound on the $\ell_\infty$ component of noise is $\epsilon$ and the size of domain is  $N$, then the error in recovered function is $\sqrt{N} \cdot \epsilon$.
}

Finally, while not the principal thrust of our paper, the problem of recovering a low-degree polynomial under adversarial outlier noise has  been studied in both   machine learning and  theoretical computer science~\cite{xu2009robust, bhatia2017consistent, herman2010general, klivans2018efficient}. In particular, \cite{klivans2018efficient} used the so-called ``Sum-of-squares" algorithm in conjunction with hypercontractivity type results 
(similar to us) to design an algorithm for recovering low-degree polynomials in presence of outlier noise over $\{0,1\}^n$. While their algorithm is robust to an even stronger notion of noise (they refer to it as \emph{nasty noise}), the precise fraction of outliers that can be tolerated for degree-$d$ polynomials over the cube is not explicit from the theorem statements in that paper (though we expect it to be qualitatively similar to ours). The algorithmic machinery is significantly different -- ours based on linear programming while theirs is based on the sum of squares method.

\paragraph{Organization.} We provide a proof
overview in Section~\ref{sec:proof_over}. In Section~\ref{sec:preli}, we introduce basic tools and notations. Section~\ref{sec:linear_program_dec} describes the principal algorithmic tool -- namely, 
a linear program whose running time is $\poly(N)$ where $N$ is the size of the set of all possible Fourier characters. 
Next we prove Theorems \ref{thm:inf_boolean_rand_out} and  \ref{thm:inf_periodic_rand_out} in Sections~\ref{sec:sparse_FFT_random_outlier} and~\ref{sec:periodic_sparse_FFT_rand} separately. Finally,
we discuss the recovery under adversarial outliers over the Boolean cube in Section~\ref{sec:adv_outlier_boolean} and torus in Section~\ref{sec:periodic_adv}.

%% file: proof_sketchv1.tex
\section{Proof Overview}\label{sec:proof_over}
In this section, we sketch the proof of our main results, i.e., Theorem~\ref{thm:inf_boolean_rand_out} and Theorem~\ref{thm:inf_periodic_rand_out}. While the domains in Theorem~\ref{thm:inf_boolean_rand_out} and Theorem~\ref{thm:inf_periodic_rand_out} are different, 
the big picture algorithmic idea is the same in both. So, for the proof sketch below, the reader can assume the domain is either the torus or the hypercube (whichever is more convenient to the reader). The domain specific ideas are highlighted whenever necessary. In this proof sketch, we assume that the reader is familiar with the basics of Fourier analysis over the domains $\{0,1\}^n$ and the torus $[0,1)$. Otherwise, we suggest reading Section~\ref{sec:preli} which discusses the basic notions of Fourier analysis over these domains. 


\paragraph*{Known support case:} Let us begin with a simple case, namely that the Fourier transform of $f$, i.e., $\wh{f}$ is $k$-sparse and further the algorithm is given the characters in the support of $\wh{f}$, say $\{\chi_1, \ldots, \chi_k\}$. Thus, the target function $f$ lies in $\mathsf{span}\{\chi_1, \ldots, \chi_k\}$ but the corresponding coefficients are unknown. Let the algorithm query the oracle at points $x_1,\ldots, x_m$ and let the observations be $y(x_1), \ldots,y(x_m)$. Recovering $f$ from the noisy observations $y(x_1), \ldots,y(x_m)$ is now essentially a case of \emph{linear regression with outlier noise}.  While presence of outlier noise makes the problem NP-hard in the worst case, as we will see, the problem is significantly more tractable when the location of the outliers is randomly distributed. In particular, note that the $\ell_1$ regression returns $g$ such that 
\begin{equation}\label{eq:ell_1_reg}
g=\underset{g \in \mathsf{span}\{\chi_1,\ldots,\chi_k\}}{\arg\min} \left\{ \sum_{i=1}^m |y(x_i)-g(x_i)| \right\}
\end{equation}
Now, suppose $S \subset [m]$ denote the points corrupted by outliers and further, it satisfies 
\begin{equation}\label{eq:prop_outliers}
\sum_{i \in S} |h(x_i)| < \sum_{i \notin S} |h(x_i)| \text{ for any } h \in \mathsf{span}\{\chi_1,\ldots,\chi_k\}. 
\end{equation}
Then, applying (\ref{eq:prop_outliers}) to $h= f-g$, it follows that the output $g$ is close to $f$ (this is explained in more detail in Section~\ref{sec:linear_program_dec}). As we will see, when the outliers are randomly distributed, (\ref{eq:prop_outliers}) holds with high probability. 
We now move to the case when the characters in the support of $f$, i.e., $\{\chi_1, \ldots, \chi_k\}$ are not known to the algorithm.
We now introduce a couple of notations (useful for the rest of this section): 
\begin{enumerate}
\item  $N$ will be the number of possible Fourier characters, which is $2^n$ in the Boolean cube $\{0,1\}^n$ and $2F+1$ in the torus $[0,1)$ with bandlimit $F$.
\item Given a function $h$, we use $\|\wh{h}\|_p$ to denote the $\ell_p$ norm on the coefficient vector, i.e., $(\sum_{\xi} |\wh{h}(\xi)|^p )^{1/p}$. Thus, $\|\wh{h}\|_0$ denotes its sparsity after the Fourier transform.
\end{enumerate}

\paragraph{Recovery under random outliers.} When each $x_i$ is uniformly sampled from the domain $D$ and picked in $S$ independently with probability $\rho$,
then both $S$ and $[m]\setminus S$  are random subsets of $D$  of size $\mathsf{Bin}(m,\rho)$ and $\mathsf{Bin}(m, 1-\rho)$ respectively. Using an $\ell_1$ concentration inequality (we give the precise statement later), it follows that for $m'= \rho m$ (when $m$ is large enough), 
\begin{equation}\label{eq:prop_outliers_2k}
\sum_{i \in [m']} |h(x_i)| \approx m' \cdot \E_{x \sim D}\bigg[ |h(x)| \bigg] \text{ for any } h \text{ with } \|\wh{h}\|_0 \le 2k.
\end{equation}
This immediately implies (\ref{eq:prop_outliers})  (where $h=f-g$ is $2k$-Fourier-sparse) and shows that $g$ defined as 
\begin{equation}\label{eq:ell_1_reg_k}
g=\underset{\|\wh{g}\|_0 = k}{\arg\min} \left\{ \sum_{i=1}^m |y(x_i)-g(x_i)| \right\}, 
\end{equation} 
is close to $f$. {In fact, Talagrand \cite{Talagrand90} and Cohen and Peng \cite{CohenPeng15} show that if the $2k$ characters in the support of $h$ are known and fixed, then (\ref{eq:prop_outliers_2k}) holds with probability $1-\gamma$ once $m = \Theta(k \log (k/\gamma))$. By applying a union bound over all subsets of $2k$ characters among all $N$ possible characters, it
follows that (\ref{eq:prop_outliers_2k}) holds as long as   $m$ is chosen to be $m= O(k^2 \log N)$. }





While this gives a statistically efficient algorithm to learn $f$, algorithmically, one needs to go over all $k$-subsets of $N$ characters. This means the time complexity blows up to $\approx N^k$. 
 In the rest of this discussion, we first outline an algorithm with running time $\poly(N)$  and then sketch an improvement to $\poly(k, \log N)$ for the hypercube and the torus. While attaining a running time of $\poly(N)$ does not rely on domain specific ideas, improvement to $\poly(k, \log N)$  relies on specific properties of the hypercube and the torus. 

\paragraph{Linear program with running time $\poly(N)$.} The high level idea to obtain a $\poly(N)$ running time is to replace the $\ell_0$ constraint in (\ref{eq:ell_1_reg_k}) with a $\ell_1$ constraint (which can be solved using linear programming). This is similar to the use of $\ell_1$ minimization in compressive sensing~\cite{CRT06}. In particular, suppose the algorithm is given an estimate of $\Delta=\sum_{i=1}^m |y(x_i) - f(x_i)|$  (we discuss how to get rid of this assumption later). Then, we consider the following $\ell_1$ relaxation of (\ref{eq:ell_1_reg_k}). 
\begin{equation}\label{eq:L1_min_outliers}
g=\argmin \big\{ \|\wh{g}\|_1 \big\} \text{ subject to } \sum_{i=1}^m |g(x_i)-y(x_i)| \le \Delta.
\end{equation}
The above minimization problem can be easily reformulated as a linear program and thus solved in time $\poly(N)$. Similar to $\ell_1$ relaxations used in the context of compressed sensing, we want to show that a solution $g$ to  (\ref{eq:L1_min_outliers})  is close to $f$ -- however, there is a crucial difficulty in doing this which we explain now.

In compressed sensing, we find a sparse solution under bounded $\ell_2$ noise (which is not the case here) by considering the relaxation 
\begin{equation}\label{eq:L1_min}
g=\argmin \bigg\{ \|\wh{g}\|_1 \bigg\} \text{ subject to } \sum_{i=1}^m |g(x_i)-y(x_i)|^2 \le \Delta,
\end{equation}
where $\Delta = \sum_{i=1}^m |y(x_i)-f(x_i)|^2$. The  argument to show that this relaxation returns $g$ close to $f$ relies on two crucial facts. The first is that 
\begin{equation}\label{eq:prop_outliers_2k_L2}
\sum_{i \in m'} |h(x_i)|^2 \approx m' \cdot \E_{x \sim D}\bigg[ |h(x)|^2 \bigg] \text{ for any } h \text{ with } \|\wh{h}\|_0 \le 2k. 
\end{equation}
The second is Plancherel's identity which states that the $\ell_2$ norm of any function and its Fourier transform are the same. Turning to the relaxation in (\ref{eq:L1_min_outliers}), while (\ref{eq:prop_outliers_2k}) can substitute for (\ref{eq:prop_outliers_2k_L2}), there is no analogue of Plancherel's identity for $\ell_1$ norm. 
Put differently, two sparse functions may have the same $\ell_1$ norm for the Fourier spectrum but very different $\ell_1$ norms (in the function space). As an example, consider the $n$-dimensional Boolean cube and the functions $h_1$ and $h_2$ defined as follows: 
\[
h_1(x) =1  \quad \textrm{and} \ \ \ h_2(x)=  \overset{\log_2 (2k)}{\underset{i=1}{\text{AND}}} x_i; 
\]
Observe that $\Vert \widehat{h_1} \Vert_0 \le 2k$, $\Vert \widehat{h_2} \Vert_0 \le 2k$ and $\Vert \widehat{h_1} \Vert1 = \Vert \widehat{h_2} \Vert_1 = 1$. However, $\mathbf{E}[|h_1|] =2k \cdot \mathbf{E}[|h_2|]$. 

To circumvent this issue, we adopt a more direct approach to show why (\ref{eq:L1_min_outliers}) returns a solution close to $f$. In particular, consider the function $h = f-g$. While $h$ is not necessarily $2k$ Fourier sparse,  $\wh{h}=\wh{f}-\wh{g}$ satisfies $\|\wh{h}\|_1 \le 2\sqrt{k} \cdot \|\wh{h}\|_2$. This relies on using  $\min \|\wh{g}\|_1$ as the objective function and is indeed different from  the objective function used in \cite{KP19} for Gaussian measurements.
Let us now define the set $\FF$ as
$$
\FF=\left\{f \bigg| \|\wh{f}\|_1 \le 2 \sqrt{k} \cdot \|\wh{f}\|_2 \right\}.
$$
Observe that $h \in \FF$. We strengthen (\ref{eq:prop_outliers_2k}) to show that 
\begin{equation}~\label{eq:FF-outliers}
\sum_{i \in [m']} |h(x_i)| \approx m' \cdot \E_{x \sim D}\bigg[ |h(x)| \bigg] \text{ for any }  h \in \FF,   
\end{equation}
where $m'=\tilde{O}(k^2 \log N)$ which in turn yields that the output $g$ of \eqref{eq:L1_min_outliers}  is close to the sparse function $f$.

Our proof of the $\ell_1$ concentration for $\FF$ crucially relies on the chaining argument from high dimensional probability theory~\cite{LTbook,RV08}. In particular, we observe that the chaining argument (relying on Maurey's empirical method) by Rudelson and Vershynin~\cite{RV08} (which they use to prove the restricted isometry property for Fourier-sparse functions) easily extends to give a $\ell_2$ concentration for the class $\FF$. This in turn allows us to prove the $\ell_1$ concentration for the class $\FF$. The details of this $\ell_1$ concentration inequality are technical and along with the description of the algorithm, are deferred to Section~\ref{sec:linear_program_dec}. {The intuition is that we lose a factor $N$ on the sample complexity if we apply a union bound directly over a net in $\FF$ because of the following type of vectors $$\bigg( \underbrace{\frac{\pm 1}{\sqrt{k}},\ldots,\frac{\pm 1}{\sqrt{k}}}_{k}, \underbrace{\pm \sqrt{k}/N, \ldots, \pm \sqrt{k}/N}_{N-k} \bigg) \in \FF.$$ 
Note that there are $2^N$ such vectors leading to a sample overhead of a factor of $N$. Instead the chaining argument allows us to save this factor of $N$ -- this is done by considering a sequence of  
nets and bounding the covering number at various radii.}

\ignore
{First of all, functions in $\FF$ are not covered by Fourier-sparse functions in $\ell_1$ distance such as $\wh{h}=(\underbrace{1,\ldots,1}_{k},\underbrace{\frac{k}{n-k},\ldots,\frac{k}{n-k}}_{n-k})$. However, we can still show $h \in \FF$ behaves like a Fourier-sparse function:
$$ 
\sup_x \bigg\{ |h(x)| \bigg\} \le 4k \cdot \E_x\bigg[|h(x)|\bigg] \text{ and } \E_x\bigg[|h(x)|\bigg] \ge \frac{\|\wh{h}\|_2}{2\sqrt{k}} \quad \forall h \in \FF.
$$ 
So we prove it using tools from high dimensional probability theory \cite{LTbook,RV08}. We notice the chaining argument using Maurey's empirical method by Rudelson and Vershynin \cite{RV08} proves not only the restricted isometry property for $2k$-Fourier-sparse functions but also the $\ell_2$ concentration for our relaxation $\FF$. This lets us to apply the chaining argument for $\ell_1$ concentration of $\FF$. We describe our algorithm and show its correctness in 
}

\paragraph{Sparse Fourier transform with running time $\poly(k,\log N)$.} Next we discuss how to reduce the dependence of the running time on $N$ from $\poly(N)$ to $\poly (\log N)$. The main idea is to use the sparse Fourier transform algorithms  \cite{GoldreichLevin:89, gilbert2002near,AGS03,GMS05,Iwen2010,HIKP12b} to do a domain reduction -- e.g., for the Boolean cube, using the ideas of \cite{GoldreichLevin:89, AGS03}, we can effective reduce the ambient domain from $\{0,1\}^n$ to $\{0,1\}^{O(\log k)}$. The sparse Fourier transform algorithms in literature fail to tolerate random outlier noise once $\rho = \Omega(1/k)$ (whereas we want to tolerate $\rho \rightarrow 1/2$). We circumvent this by using the linear program described above  (i.e., \eqref{eq:L1_min_outliers}) on the reduced domain which allows us to tolerate any $\rho<1/2$ fraction of random outliers. Section~\ref{sec:sparse_FFT_random_outlier} gives the details of this algorithm for the Boolean cube and Section~\ref{sec:periodic_sparse_FFT_rand} gives the details for functions over the torus $[0,1)$.



\paragraph{Recovery under adversarial outliers.} As we mentioned earlier, in the adversarial outlier noise model,  $\rho=1/k$ is an information theoretic limit on the fraction of outliers which can be tolerated when recovering $k$-Fourier sparse functions. However, assuming some further structural restrictions on the functions, we are able to circumvent this limit. In particular, for degree-$d$ polynomials over the Boolean cube, we are able to tolerate $\rho \approx \frac{1}{4\cdot 3^{2d}}$ -- as opposed to $\rho \approx 1/n^d$ which we get by just observing that degree-$d$ polynomials over $\{0,1\}^n$ are $k$-Fourier sparse for $k \approx n^d$.

To obtain this bound, we appeal to the anti-concentration of low degree polynomials: Namely, when $h$ is a degree-$d$ polynomial, then $\E[|h|] \ge 3^{-d} \cdot \E[|h|^2]$. On the other hand, once $m$ is large enough, it  easily follows that 
$$
\sum_{i=1}^m |h(x_i)| \approx m \cdot \E[|h|] \text{ and } \sum_{i=1}^m |h(x_i)|^2 \approx m \cdot \E[|h|^2] \text{ for any $h$ of degree } d. 
$$ 
By plugging the anti-concentration of degree-$d$ polynomials into the above relation, we get that for any set $S$ of size smaller than $\frac{m}{4 \cdot 3^{2d}}$, $\sum_{i \in S} |h(x_i)| < \frac{1}{2} \sum_{i=1}^m |h(x_i)|$ (see Theorem~\ref{thm:euclidean_section}). This easily shows that the linear program defined in (\ref{eq:ell_1_reg}) can tolerate $\rho$ up to $\frac{1}{4\cdot 3^{2d}}$ fraction of outliers (where $\chi_1, \ldots, \chi_k$ are all the monomials of degree at most $d$ over $\{0,1\}^n$). The details are  described in 
Section~\ref{sec:adv_outlier_boolean}.

Finally, for the torus $[0,1)$, we show that it is possible to beat $\rho = 1/k$ bound for $k$-Fourier sparse functions (and in fact get any $\rho<1/2$) when all the non-zero Fourier coefficients $\wh{f}(\xi)$ are integral multiples of some given number $\eta$. The proof of this relies on an anti-concentration bound for such functions which relies on techniques from complex analysis. Elaborating a little more, we use simple properties of harmonic functions to show that the \emph{radius of anti-concentration} of a polynomial 
can be lower bounded just in terms of $\eta$ where $\eta$ is the smallest non-zero coefficient of the polynomial. Applying this to the function $h=f-g$ where $g$ is the output of the linear program defined by (\ref{eq:ell_1_reg}) yields the final result. 
As opposed to result for random outliers, 
this algorithm has a running time dependent on $\poly(F)$ where $[-F,F]$ is the bandlimit and thus is not efficient in terms of the running time. 
The details of this result appears in Section~\ref{sec:periodic_adv}.


%% file: preli.tex
\section{Preliminaries}\label{sec:preli}
\paragraph*{Notations:} We use $[n]$ to denote $\{1,2,\ldots,n\}$. Given a subset $S$ and a ground set $U$, we use $\overline{S}$ to denote its complement $U \setminus S$.

Given a vector $v \in \mathbb{R}^m$, we use $\|v\|_p$ to denote its $\ell_p$ norm $(\sum_{i=1}^m |v(i)|^p)^{1/p}$. For a subset $S \subseteq [m]$, we use $v_S$ to denote the vector restricted to $S$, i.e., $v_S(i)=v(i) \cdot 1_{i \in S}$. 

We use $\wt{O}(T)$ to hide terms which are polynomial in $\log T$. We use $X \lesssim Y$ to denote that for some constant $C$, $X \le C \cdot Y$. Likewise, $X \gtrsim Y$ denotes that there is a constant $C$ such that $X \ge C \cdot Y$. Finally, we  use $\exp(-n)$ to denote a quantity exponentially small  in $n$, i.e., $C^{-n}$ for some $C>1$.

Finally, if $f$ is the unknown target function and $x$ is any point in the domain, then we let $y(x)$ denote the observation at $x$ and $e(x)$ denote the noise -- i.e, $y(x) =f(x) + e(x)$. For $(\rho,\epsilon)$ outlier noise, note that if $x$ is not an outlier, then $|e(x)|\le \epsilon$. If we are in the random outlier noise model, then each observation is an outlier with probability $\rho$ (independently at random) whereas in the adversarial noise model, $\rho$ is an upper bound on the fraction  of outliers. 
{{Without loss of generality, we assume $\|y\|_{\infty} \le \poly(k)$ in the rest of this work.}}

\paragraph{Fourier transform.} We begin by defining Fourier transform over 
 the Boolean cube $\{0,1\}^n$. Let $f: \{0,1\}^n \rightarrow \mathbb{C}$  and let us define $\chi_{\xi}(x) =  (-1)^{\langle \xi,x\rangle}$ where $\xi \in 
 \fF_2^n$. We 
 define $\left\{ \chi_\xi(\cdot) \big| \xi \in \fF_2^n \right\}$ to be the set of  characters over $\fF_2^n$. For each such $\xi$, define the corresponding Fourier coefficient as $\wh{f}(\xi)=\underset{x \sim \{0,1\}^n}{\E}[f(x) \cdot (-1)^{\langle \xi,x\rangle}]$.  From the definition of Fourier coefficients, it easily follows that $f(x)=\sum_{\xi} \wh{f}(\xi) (-1)^{\langle \xi,x\rangle}$.

Given $\xi \in \{0,1\}^n$, we define the degree of the character $(-1)^{\langle \xi,x \rangle}$ to be its Hamming weight --- $|\xi| \overset{\text{def}}{=} \sum_i \xi(i)$. Given $f(x)=\sum_{\xi} \wh{f}(\xi) (-1)^{\langle \xi,x\rangle}$, we define the degree of $f$ to be $\max_{\xi: \wh{f}(\xi) \neq 0} \{ |\xi| \}$. 
Alternately, this is the same as the degree of $f$ when expressed as a multilinear polynomial in the variables 
$x_1, \ldots, x_n$. 

We now turn to Fourier analysis over the torus $\mathbb{R}/\mathbb{Z}$. In particular, functions $f$ in this domain can either be identified with periodic functions over $\mathbb{R}$ -- i.e., $f(x) = f(x+z)$ for any $x \in \mathbb{R}$ and $z \in \mathbb{Z}$. Alternately, this is the same as the space of  functions  $f: [0,1) \rightarrow \mathbb{C}$. For any $\xi \in \mathbb{Z}$, we define the character $\chi_\xi: [0,1) \rightarrow \mathbb{C}$ as $\chi_\xi(t) = e^{2 \pi \bi \cdot \xi t}$. The corresponding Fourier coefficient of $f$, denoted by $\wh{f}(\xi)$ is given by $\wh{f}(\xi)=\int_0^1 f(t) \cdot e^{-2 \pi \bi \cdot \xi t} \mathrm{d} t$. Assuming $f$ is both $\ell_1$ and $\ell_2$ integrable, it also follows that  
$f(t)=\sum_{\xi} \wh{f}(\xi) e^{2 \pi \bi \cdot \xi t}$.  Since all functions in this paper will be bounded (and hence $\ell_p$ integrable for any $p \ge 0$), we will henceforth not state this condition explicitly.

Observe that unlike $\{0,1\}^n$, the number of characters (and hence the Fourier coefficients) is infinite. In this paper, we will be interested in so-called \emph{bandlimited functions}. In other words, the algorithm will be given $F$ such that the target function $f$ has all its non-zero Fourier coefficients $\xi$ lying in the set $[-F, F] \cap \mathbb{Z}$. For an arbitrary function $g$, we define its \emph{bandlimited spectrum} 
(defined by $F$) as its Fourier coefficients $\wh{g}(\xi)$ where $\xi \in [-F, F] \cap \mathbb{Z}$.


We will also use the Fourier transform over the  cyclic group  $\mathbb{Z}_n$. The characters are given by $\chi_{\xi}(x) = e^{2 \pi \bi \frac{\xi x}{n}}$ where $\xi \in \mathbb{Z}_n$. For $f:\mathbb{Z}_n \rightarrow \mathbb{C}$, the Fourier coefficient corresponding to $\chi_\xi$ is given by $\wh{f}(\xi)=\underset{x \sim [0,1)}{\E}[f(x) e^{-2 \pi \bi \frac{\xi x}{n}}]$.

Over all these domains, we will define two fundamental binary operations between functions. For $f$ and $g$, we define the dot product $(f \cdot g)(x)=f(x) \cdot g(x)$. Similarly, the convolution $(f * g)(x)=\sum_{x'} f(x') \cdot g(x-x')$. Two fundamental properties of these operations are: 
\begin{enumerate}
\item $\wh{f \cdot g}=\wh{f} * \wh{g}$ and $\wh{f \ast g}=\wh{f} \cdot \wh{g}$. 
\item \textbf{Parseval’s identity:} $
\E_x [|f(x)|^2] = \sum_{\xi} |\wh{f}(\xi)|^2.
$
\end{enumerate}

\paragraph{Facts about the Gaussian variables.} We always use $N(0,1)$ to denote the standard Gaussian random variable and use the following concentration bound on Gaussian random variables \cite{LTbook}.
\begin{lemma}\label{lem:Gaussian_variables}
Given any $n$ Gaussian random variables $G_1,\cdots,G_n$ (not necessarily independent) where each $G_i$ has expectation 0 and variance $\sigma_i^2$, 
\[
\E \big[ \max_{i\in [n]} |G_i| \big] \lesssim \sqrt{\log n} \cdot \max_{i\in [n]} \big\{ \sigma_i \big\}.
\]
\end{lemma}

%% file: LP_decoding-ad.tex
\section{Linear Program Decoding for random outliers}\label{sec:linear_program_dec}
Given a domain $D$, a set $T$ of Fourier characters of $D$ and an oracle (with random outlier noise) to a  function $f$ supported 
on $k$ of these $T$ characters, we provide an algorithm with running time polynomial in $|T|$ which recovers $f$ (with small error). Note that while we do not give the definition of Fourier characters for an arbitrary domain $D$, for the purposes of this section, we just use two properties: (i) The set of Fourier characters is a set of orthonormal functions (with respect to the uniform measure on $D$). (ii) The $\ell_\infty$ norm of any Fourier character is $1$.  
\begin{theorem}\label{thm:LP_guarantee_boolean_cube}
There is an algorithm which when given as input, sparsity parameter $k$, domain $D$, a set $T$ of Fourier characters over $D$, failure probability $\gamma$, parameters $\delta>0$ and $\eta>0$, and an oracle to $f=\sum_{j=1}^k \wh{f}(\xi_j) \cdot \chi_{j}$ (with each $\chi_j \in T$ and  $|\wh{f}(\xi_j)| \ge \eta$) with $( \frac{1}{2}-\delta, \eps)$ random outlier noise (where $\eps \lesssim \eta \cdot \delta$), makes $\tilde{O}(k^2 \log |T| \log \frac{1}{\gamma}/\delta^2)$ queries and runs in time $\poly(|T|, 1/\delta,1/\eta,\log \frac{1}{\gamma})$ and 
outputs $g$ satisfying 
$$
\supp(\wh{g})=\supp(\wh{f}) \text{ and } \E_x[|g(x)-f(x)|] \lesssim \frac{\eps}{\delta},$$   with probability  $1-\gamma.$ In particular, this implies $|\wh{g}(\xi)-\wh{f}(\xi)| \lesssim \frac{\epsilon}{\delta}$ for every  $\xi$.
\end{theorem}
The algorithm is described in Algorithm~\ref{alg:LP_decoding_boolean} and we prove its correctness in the rest of this section. To do so, we will consider the following optimization problem (which can be easily formulated as a linear program). In particular, suppose the observations at points $x_1, \ldots, x_m$ are $y(x_1), \ldots, y(x_m)$ respectively. Further, for $1 \le i \le m$, define $e(x_i) = y(x_i) -f(x_i)$ be the noise at point $i$ and let $\Delta$ be an estimate of this noise, i.e, $\sum_{i=1}^m |e(x_i)|$. Then, the optimization problem (where $\{g(x)\}_{x \in D}$ are the unknowns) is 
\begin{equation}\label{eq:LP_decode_sparseFFT}
\min \|\wh{g}\|_1 \quad \text{ subject to } \sum_{i=1}^m |g(x_i)-y(x_i)| \le \Delta \text{ and } g \in \text{span}\{T\}.
\end{equation}
It is easy to see that (\ref{eq:LP_decode_sparseFFT}) can be formulated as a linear program. Algorithm~\ref{alg:LP_decoding_boolean} is described next. 
\begin{algorithm}[H]
\caption{Linear Program Decoding for sparse FFT}\label{alg:LP_decoding_boolean}
\begin{algorithmic}[1]
\Procedure{\textsc{LinearDecodingSparseFFT}}{$y,T,k,\gamma,\delta,\eta$}
\State $m:=\tilde{O}\bigg( k^2 \log |T| \log \frac{1}{\gamma} /\delta^2 \bigg)$ 
\State Sample $m$ random points $x_1,\ldots,x_m$ and let $y(x_1),\ldots,y(x_m)$ be the corresponding observations. 
\State $\sigma:=m \cdot \frac{\eta \cdot \delta}{100}$
\For{$\Delta$ from 0 to $\sum_i |y(x_i)| + \eta m$ with gap $\sigma$}
\State Solve the linear program \eqref{eq:LP_decode_sparseFFT} to obtain $g$ given $\Delta$, $T$, and $x_1,\ldots,x_m$ with the corresponding observations $y(x_1),\ldots,y(x_m)$.
\State Set $S_{\Delta}=\{\textit{The $k$ characters in $g$ with the largest absolute coefficients}\}$.
\State Let $g_{\Delta}=\underset{h \in \text{span}(S_{\Delta}) }{\arg\min} \sum_{i=1}^m |h(x_i)-y(x_i)|$.
\EndFor
\\
\Return $g=\underset{g_{\Delta}}{\arg\min} \sum_{i=1}^m |g_{\Delta}(x_i)-y(x_i)|$.
\EndProcedure
\end{algorithmic}
\end{algorithm}
First notice that under the assumption $\|y\|_{\infty} \le \poly(k)$, our algorithm runs in time $\poly(k,|T|,1/\eta,1/\delta)$.
In the rest of this section, we only consider $h:D \rightarrow \mathbb{R}$ whose Fourier transform is supported on  $T$ and let $\|\wh{h}\|_1$ and $\|\wh{h}\|_2$ denote $\sum_{\chi \in T} |\wh{h}(\chi)|$ and $\big( \sum_{\chi \in T} |\wh{h}(\chi)|^2 \big)^{1/2}$ separately. We next show the following  guarantee for the LP \eqref{eq:LP_decode_sparseFFT}.
\begin{lemma}\label{lem:guarantee_LP}
For parameters $\gamma>0,\delta>0,$ and $\eps>0$ , define $\rho=1/2-\delta$ and $m=\tilde{O} \bigg( k^2 \log |T| \log \frac{1}{\gamma}/\delta^2 \bigg)$. Let $f: D \rightarrow \mathbb{R}$ be such that $\|\wh{f}\|_0 \le k$ and $y(x)$ be the output of an oracle for $f(x)$ with $(\rho,\epsilon)$ random outlier noise. For any query point $x_i$, define $y(x_i) = f(x_i) + e(x_i)$ and let $\Delta$ be an 
 upper bound on $\sum_{i=1}^m |e(x_i)|$ and  $\sigma_{\gap}$ be defined as $\sigma_{\gap}\overset{\text{def}}{=}\Delta-\sum_{i=1}^m |e(x_i)|$. If the points $x_1, \ldots, x_m$ are chosen at random, 
then the linear program \eqref{eq:LP_decode_sparseFFT} returns $g$ satisfying 
$$
\E[|f-g|] \lesssim \frac{\sigma_{\gap}}{m \cdot \delta} + \frac{\eps}{\delta} \text{ with probability }1-\gamma.
$$ 
This implies every coefficient $\xi$, $|\wh{f}(\xi)-\wh{g}(\xi)| \lesssim \frac{\sigma_{\gap}}{m \cdot \delta} + \frac{\eps}{\delta}$.
\end{lemma}
Note that the output $g$ of the linear program \eqref{eq:LP_decode_sparseFFT} is not necessarily Fourier sparse. Next, we have the following lemma whose proof is quite similar to that of Lemma~\ref{lem:guarantee_LP}. Note that unlike $g$ in Lemma~\ref{lem:guarantee_LP}, the function $g$ in Lemma~\ref{cor:correctness_over_k_sparse} is not efficiently computable. 
\begin{lemma}\label{cor:correctness_over_k_sparse}
For parameters $\gamma>0,\delta>0,$ and $\eps>0$ , define $\rho=1/2-\delta$ and $m=\tilde{O} \bigg( k^2 \log |T| \log \frac{1}{\gamma}/\delta^2 \bigg)$.  For $f: D \rightarrow \mathbb{R}$, let $f=\sum_{j=1}^k \wh{f}(\xi_j) \cdot \chi_{j}$ such that  $|\wh{f}(\xi_j)| \ge \frac{5\epsilon}{\delta}$ for all $1 \le i \le k$. For any point $x$, let $y(x)$ be the output of an oracle for $f$ with $(\rho,\epsilon)$ random outlier noise. If the points $x_1, \ldots, x_m$ are chosen at random, 
 then with probability $1-\gamma$, $g=\underset{h:\|\wh{h}\|_0 \le k}{\arg\min} \sum_{i=1}^m |y(x_i)-h(x_i)|$ satisfies
$$
\supp(\wh{g})=\supp(\wh{f}) \text{ and } \E_x[|g(x)-f(x)|] \lesssim \frac{\eps}{\delta}.
$$
In particular, this implies $|\wh{g}(\xi)-\wh{f}(\xi)| \lesssim \frac{\epsilon}{\delta}$ for every $\xi$
\end{lemma}

The main technical tool in these proofs is a concentration bound for the following family of functions
\begin{equation}\label{eq:def_FF}
\FF=\left\{ h \big| \supp\{\wh{h}\} \subseteq T, \|\wh{h}\|_1 \le 2 \sqrt{k} \|\wh{h}\|_2 \right\},
\end{equation}
which is a \emph{relaxation} of the family of $2k$-Fourier-sparse functions. The next lemma gives an algorithm for the estimation of $\ell_1$ norm of $h \in \mathcal{F}$. 
\begin{claim}\label{clm:ell_1_ball}
For any $\epsilon$ and failure probability $\gamma$, there exists $m=O\bigg( k^2 \log |T| \cdot \log \frac{1}{\gamma} \cdot \frac{\log^3 \frac{k \log |T|}{\epsilon}}{ \eps^2 }\bigg)$ such that for $m$ random points $x_1,\ldots,x_m \in D$, with probability $1-\gamma$,
$$
\sum_{i=1}^m |h(x_i)| = (1 \pm \epsilon) \cdot m \cdot \E_{x \sim D}\big[ |h(x)| \big] \text{ for any } h \in \mathcal{F}. 
$$ 
\end{claim}

We finish the proof of Theorem~\ref{thm:LP_guarantee_boolean_cube} using Lemma~\ref{lem:guarantee_LP} and Lemma~\ref{cor:correctness_over_k_sparse} here. We defer the proofs of Lemma~\ref{lem:guarantee_LP} and Lemma~\ref{cor:correctness_over_k_sparse} to Section~\ref{sec:supp_proofs_LP} and the proof of Claim~\ref{clm:ell_1_ball} to Section~\ref{sec:proof_ell_1_ball}.

\begin{proofof}{Theorem~\ref{thm:LP_guarantee_boolean_cube}}
We first show that $\sum_{i=1}^m |e(x_i)| \le \sum_i |y(x_i)| + 2\eps m$. To show this, let $y(x_i) = f(x_i) + e(x_i)$ and $S \subseteq [m]$ denote the subset of queries corrupted by outliers -- i.e., $S = \{i \le m: |e(x_i)| > \epsilon\}$. 

We observe that  with probability $1-\exp(-\delta^2 m)$, $|S| \in [1/2-5\delta/4,1/2-3\delta/4] \cdot m$. Consequently, with probability $1-\exp(-\delta^2 m)$, we also have 
\begin{align}
\sum_{i \notin S} |y(x_i)| & \ge \sum_{i \notin S} |f(x_i)| - \eps (m-|S|) \nonumber \\
& \ge (1/2+3\delta/4)m \cdot (1 - \delta/2) \E[|f(x)|] - \eps (m - |S|) \nonumber \\
& \ge m/2 \cdot \E[|f(x)|] - \eps m. \label{eq:ineq-4-1}
\end{align}
Here, the second inequality follows by applying Claim~\ref{clm:ell_1_ball} (with $\epsilon = \delta/2$ and failure probability $\gamma/2$) to $f$.  Next, observe that 
\begin{align}
\sum_{i=1}^m |e(x_i)|  &\le \sum_{i \in S} |e(x_i
)| + \sum_{i 
\not \in S} |e(x_i
)|  \nonumber \\   &\le 
 \sum_{i \in S} |e(x_i
)|  + \eps m \nonumber \ \  (\text{by definition of } S) \\  &\le \sum_{i \in S} |y(x_i
)| + \sum_{i \in S} |f(x_i) | + \eps m \nonumber \ \  (\text{triangle inequality}) \\  
&\le \sum_{i \in S} |y(x_i
)|  + m/2 \cdot \E[|f(x)|] + \eps m \nonumber  \ \ (\text{Claim~\ref{clm:ell_1_ball}})  \\ 
&\le \sum_{i=1 }^m |y(x_i
)| + 2 \eps m \ \ \text{using \eqref{eq:ineq-4-1}} \nonumber \\
\end{align}
Hence there exists $\Delta^*$ in Step~5 of Algorithm~\ref{alg:LP_decoding_boolean} such that 
$$
\sum_{i=1}^m |e(x_i)| \le \Delta^* \le \sum_{i=1}^m |e(x_i)|+\sigma.
$$ 
For such a $\Delta^*$, by Lemma~\ref{lem:guarantee_LP}, the linear program given in \eqref{eq:LP_decode_sparseFFT} returns $g$ such that  $\underset{x}{\E}[|g(x)-f(x)|] \le \eta/3$ as long as  $\eta \ge \frac{10 \eps}{\delta}$ (for our choice of $\sigma$). {This implies $|\wh{g}(\xi)-\wh{f}(\xi)| \le \eta/3$ for any $\xi$. So $|\wh{g}(\xi)| \ge 2\eta/3$ when $|\wh{f}(\xi)| \ge \eta$ and it is less than $\eta/3$ when $\xi \notin \supp(\wh{f})$. This shows that the set of the largest $k$ coefficients in $\wh{g}$ is the  same as the set  $\supp(\wh{f})$, i.e., $S_{\Delta^*}$ in Step~7 of Algorithm~\ref{alg:LP_decoding_boolean} is equal to $\supp(\wh{f})$. 

At the same time, Lemma~\ref{cor:correctness_over_k_sparse} proves the $k$-Fourier-sparse minimizer $g^*=\underset{h:\|\wh{h}\|_0 \le k}{\arg\min} \sum_{i=1}^m |y(x_i)-h(x_i)|$ has the property $\supp(\wh{g})=\supp(\wh{f})$, which is $S_{\Delta^*}$ here. Thus $g_{\Delta^*}$ calculated in Step~8 of Algorithm~\ref{alg:LP_decoding_boolean} will be the minimizer for Step~10 of the algorithm which will be the same as $g^\ast$. 
Applying Lemma~\ref{cor:correctness_over_k_sparse} finishes the proof.} 
\end{proofof}

\subsection{Proofs of Lemma~\ref{lem:guarantee_LP} and Lemma~\ref{cor:correctness_over_k_sparse}}\label{sec:supp_proofs_LP}
We first prove Lemma~\ref{lem:guarantee_LP} to show the guarantee of the linear program defined in ~\eqref{eq:LP_decode_sparseFFT}.

\begin{proofof}{Lemma~\ref{lem:guarantee_LP}}
Let  $h$ denote $f-g$ and $W$ denote $\supp(\wh{f})$. Since $\|\wh{g}\|_1 \le \|\wh{f}\|_1$, we have
\begin{align*}
\|(\wh{f})_W \|_1 = \|\wh{f}\|_1 & \ge \|\wh{g}\|_1 \\
& \ge \|(\wh{f})_W \|_1 - \|(\wh{h})_W\|_1 + \|(\wh{g})_{\overline{W}}\|_1
\end{align*}
Since $(\wh{g})_{\overline{W}}=(\wh{h})_{\overline{W}}$, This implies $\|(\wh{h})_W\|_1 \ge \|(\wh{h})_{\overline{W}}\|_1$. So $h$ is in the family 
$$
\mathcal{F}_0=\left\{ \|(\wh{h})_W\|_1 \ge \|(\wh{h})_{\overline{W}}\|_1 \bigg| \forall W \in {T \choose k} \right\}.
$$
At the same time, by $\|\wh{h}\|_2 \ge \|(\wh{h})_W\|_2 \ge \frac{1}{\sqrt{|W|}} \|(\wh{h})_W\|_1$ and $|W| \le k$, we have $\FF_0 \subseteq \FF$ for 
$$
\FF=\left\{ h \bigg| \|\wh{h}\|_1 \le 2 \sqrt{k} \cdot \|\wh{h}\|_2 \right\}.
$$
Let $S$ denote the subset of $[m]$ containing the outliers. With probability $1-\exp(-\delta^2 n)$, $|S| \le (1/2-3\delta/4)m$. As $g$ is a solution to the linear program~\eqref{eq:LP_decode_sparseFFT}, we get 
\begin{align*}
\Delta & \ge \sum_i |f(x_i)+e(x_i) - g(x_i)| \\
& = \sum_i |h(x_i)+e(x_i)| \\
& \ge \sum_{i \in S} \left( |e(x_i)|-|h(x_i)| \right) + \sum_{i \in \overline{S}} \left( |h(x_i)|-|e(x_i)| \right)
\end{align*}
Since  $\sigma\overset{\text{def}}{=}\Delta-\sum_{i=1}^m |e(x_i)|$, this shows 
\begin{equation}\label{eq:LP_guarantee_under_estimation}
\sigma + 2\sum_{i \in \overline{S}} |e(x_i)|+ \sum_{i \in S} |h(x_i)| \ge \sum_{i \in \overline{S}} |h(x_i)|.
\end{equation}

At the same time, we choose the failure probability in Claim~\ref{clm:ell_1_ball} to be $\gamma/2$ such that with probability $1-\gamma$, we have both
$$
\sum_{i \in S} |h(x_i)| \le (1+\delta/4) \cdot |S| \cdot \E_{x \sim D} \big[ |h(x)| \big]
$$
and
$$
\sum_{i \in \overline{S}} |h(x_i)| \ge (1 - \delta/4) \cdot |\overline{S}| \cdot \E_{x \sim D} \big[ |h(x)| \big].
$$
Plugging these two bounds into~\eqref{eq:LP_guarantee_under_estimation} with $|S| \le (1/2-3\delta/4)m$, we have
$$
\sigma + 2\sum_{i \in \overline{S}} |e(x_i)| \ge 0.9 \delta \cdot m \cdot  \E_{x \sim D} \big[ |h(x)| \big].
$$
Because $|e(x_i)| \le \eps$ for any $i \in \overline{S}$, we have $$
\E_{x \sim D} [|h(x)|] \le 2 \frac{\sigma}{\delta \cdot m} + 3\frac{\eps}{\delta}.
$$
Thus, with probability $1-\gamma$ over the choice of $x_1, \ldots, x_m$, we have that 
for each $\xi$, $$|\wh{f}(\xi)-\wh{g}(\xi)|=|\wh{h}(\xi)| \le \E_{x \sim D} [|h(x)|] \le  2 \frac{\sigma}{\delta \cdot m} + 3\frac{\eps}{\delta},$$
which finishes the proof. 
\end{proofof}

Next we prove Lemma~\ref{cor:correctness_over_k_sparse} for $k$-Fourier-sparse functions, whose proof is very similar to the above proof of Lemma~\ref{lem:guarantee_LP}.

\begin{proofof}{Lemma~\ref{cor:correctness_over_k_sparse}}
Consider any $g \not =f $ (where $f$ is the true target function).  The definition of $g$ implies
$
\sum_{i=1}^m |y(x_i)-g(x_i)| < \sum_{i=1}^m |e(x_i)|.
$
Let $h$ denote $f-g$ and $S = \{i \in [m]: |y(x_i) - f(x_i)| >\eps\}$. 
We have
$$
\sum_{i=1}^m |e(x_i)+h(x_i)| < \sum_{i=1}^m |e(x_i)|.
$$ 
Now, we can lower bound the L.~H.~S.~of the above inequality as 
$$
\sum_{i=1}^m |e(x_i)+h(x_i)| \ge \sum_{i \in S} |e(x_i)|-|h(x_i)| + \sum_{i \notin S} |h(x_i)|-|e(x_i)|.
$$ Together, the last two inequalities imply 
\begin{equation}\label{eq:LP_noise_boolean}
\sum_{i \notin S} |h(x_i)| < \sum_{i \in S} |h(x_i)| + 2 \epsilon \cdot (m- |S|).
\end{equation}
Further, with probability $1-\exp(-\delta^2 \cdot m)$,
 we have $|S| \le (1/2-3\delta/4) \cdot m$ . Since $S$ is independent with $x_1,\ldots,x_m$ and their observations $y(x_1),\ldots,y(x_m)$, without loss of generality, we assume $S=\{1,2,\ldots,|S| \}$. Since $h$ is $2k$-sparse, we apply Claim~\ref{clm:ell_1_ball} twice with failure probability $\gamma/2$ to show $h \in \FF$ has 
$$
\sum_{i \in S} |h(x_i)| \le (1+\delta/4) \cdot |S| \cdot \E[|h(X)|] \text{ and } \sum_{i \notin S} |h(x_i)| \ge (1 - \delta/4) \cdot (m-|S|) \cdot \E[|h(X)|].
$$ 
Thus \eqref{eq:LP_noise_boolean} implies
$$
(1-\delta/4) \cdot (1/2+3\delta/4)m \cdot \E[|h(X)|] < (1+\delta/4) \cdot (1/2 - 3\delta/4)m \cdot \E[|h(X)|] + \epsilon \cdot (1+3\delta/2) m.
$$

Thus $\E[|h|] < \frac{2\epsilon}{\delta}$. Moreover, it implies $|\wh{h}(\xi)| \le \E[|h|]< \frac{2 \eps}{\delta}$. By the definition $h=f-g$, we have $|\wh{g}(\xi) - \wh{f}(\xi)| \le \frac{2 \eps}{\delta}$. Since each non-zero coefficient in $f$ has absolute value $\ge \frac{5\epsilon}{\delta}$, we have $|\wh{g}(\xi)|>0$ for any $\xi \in \supp(\wh{f})$. Because both $f$ and $g$ are $k$-Fourier-sparse, it shows $\supp(\wh{g})=\supp(\wh{f})$.
\end{proofof}

\subsection{Proof of Claim~\ref{clm:ell_1_ball}}\label{sec:proof_ell_1_ball}
Without loss of generality,  we can restrict proving Claim~\ref{clm:ell_1_ball} to the subset $ \mathcal{F}_1 \subsetneq \FF$
 $$
\mathcal{F}_1=\left\{ \|\wh{h}\|_1 \le 2\sqrt{k} \bigg| \|\wh{h}\|_2=1 \right\}.
$$ 
We first state two properties for $\FF_1$.
\begin{fact}\label{facts:family_FF}
For any $h \in \FF_1$, we  have
\begin{equation}\label{eq:properties_ell_1_ball}
\max_x |h(x)| \le  2 \sqrt{k} \quad \text{ and } \quad \E_{x \sim D}\big[ |h(x)| \big] \ge \frac{1}{2 \sqrt{k}}.
\end{equation}
\end{fact}
\begin{proof}
First, $\max_x |h(x)| \le \|\wh{h}\|_1 \le 2 \sqrt{k} \|\wh{h}\|_2$. Then $\E_{x \sim D}\big[ |h(x)| \big] \ge \frac{\E_{x \sim D}\big[ |h(x)|^2 \big]}{\max_x |h(x)|} \ge \frac{\|\wh{h}\|^2_2}{2 \sqrt{k} \|\wh{h}\|_2}$.
\end{proof}
We now state the main technical result for  $\FF_1$.
\begin{claim}\label{clm:ell_1_concentration}
There exists $m=O\big( k \log |T| \cdot \log^{3} \frac{k \log |T|}{\eps}/\eps^2 \big) $ such that 
\begin{equation}\label{eq:gaussian_ell_1_ball}
\E_{z_1,\ldots,z_m} \left[ \sup_{h \in \FF_1} \left| \sum_{i=1}^m |h(z_i)| - m \cdot \E_x\big[ |h(x)| \big] \right| \right] \le \epsilon \cdot m. 
\end{equation}
This implies 
$$
\E_{z_1, \ldots, z_m}\left[ \sup_{h \in \FF: \E_{x \sim D}[|h(x)|]=1} \left| \sum_{i=1}^m |h(z_i)| - m  \right| \right] \le \eps \cdot m \text{ for } m = O\bigg( \frac{k^2 \log |T|}{\eps^2} \cdot \log^{3} \frac{k \log |T|}{\eps} \bigg)
$$ by rescaling $\eps$ to $\frac{\eps}{2\sqrt{k}}$ with the lower bound of $ \E_{x \sim D} \big[ |h(x)| \big]$ in \eqref{eq:properties_ell_1_ball}. 
\end{claim}
By Markov's inequality, it implies Claim~\ref{clm:ell_1_ball} with $m = O\bigg( k^2 \log |T| \cdot \log^{3} \frac{k \log |T|}{\gamma \eps} \cdot \frac{1}{\eps^2 \gamma^2} \bigg)$ for any  $\gamma$. It is straightforward to extend Theorem 3.3 in \cite{RV08} to obtain a better dependence of $\log \frac{1}{\gamma}$ on $m$. We use McDiarmid's inequality to provide an alternative argument. 
\begin{lemma}[McDiarmid's inequality \cite{McDiarmid:89}]
Let $F$ be a function which is $c_i$-Lipschitz in the $i^{th}$ direction. In other words, 
 $$ \underset{x_1,\ldots,x_n,x'_i}{\sup} \big|F(x_1,\ldots,x_n)-F(x_1,\ldots,x_{i-1},x'_i,x_{i+1},x_n) \big| \le c_i, $$ for each $i \in [n]$. Then, for any $\eps>0$,
$$
\Pr_{z_1, \ldots, z_n \sim D}\bigg[|F(z_1,\ldots,z_n)-\E_{x}[F(x_1,\ldots,x_n)]| \ge \eps \bigg] \le 2\exp\left( -\frac{2\eps^2}{\sum_{i=1}^n c_i^2} \right).
$$
\end{lemma}

\begin{proofof}{Claim~\ref{clm:ell_1_ball}}
Given $z_1,\ldots,z_n$, we define $F(z_1,\ldots,z_n)$ to be $$
\sup_{h \in \FF: \E[|h|]=1} \left| \sum_{i=1}^m |h(z_i)| - m  \right|.
$$
Notice that for any $z_i,z'_i \in D$, Fact~\ref{facts:family_FF} shows
$$
|F(z_1,\ldots,z_n)-F(z_1,\ldots,z_{i-1},z'_i,z_{i+1},z_n)| \le \sup_{h \in \FF: \E[|h|=1]} \sup_{x} |h(x)| \le 4k.
$$
Choosing  $m=O\bigg( \log \frac{1}{\gamma} \cdot k^2 \log |T| \cdot \log^{3} \frac{k \log |T|}{\eps}/\eps^2 \bigg)$ implies (using Claim~\ref{clm:ell_1_concentration}) that 
$$
\E[F(x_1,\ldots,x_m)] \le \eps \cdot m.
$$
Then, McDiarmid's inequality implies
$$
\Pr\left[\bigg|F(x_1,\ldots,x_m) - \E[F(x_1,\ldots,x_m)] \bigg| \ge a \right] \le 2 \exp\left(-\frac{2a^2}{(4k)^2 \cdot m}\right).
$$
For $a=\eps m$, this is at most $\gamma$. We rescale $\eps$ to finish the proof.
\end{proofof}
Thus, it remains to prove Claim~\ref{clm:ell_1_concentration} -- to do this, we use a standard symmetrization and Gaussianization argument~\cite{LTbook,RW} which transforms bounding the left hand side of \eqref{eq:gaussian_ell_1_ball} to bounding a Gaussian process. In particular, we will use the following theorem (whose proof, for 
completeness, is provided in Appendix~\ref{appd:sym_gau}). 

\define{thm:sym_gauss}{Theorem}{ Let $\bX$ be a random variable, $S$ be a set and $f: S \times \supp(\bX) \rightarrow \mathbb{R}^+$ be a non-negative function. Let $\bx= (x_1, \ldots, x_n)$ and $\bx'=(x'_1, \ldots, x'_n)$ be independent draws from $\bX^n$ and $\boldsymbol{g}=(g_1,\cdots,g_n)$ be an independent draw from the $n$-dimensional standard normal $N(0,1)^n$. 
Then, 
$$
\underset{\bx}{\E}\left[\max_{\Lambda \in S} \big|\sum_{j=1}^n f(\Lambda,x_j) - \underset{\bx'}{\E}[\sum_{j=1}^n f(\Lambda,x'_j)] \big| \right] \le \sqrt{2\pi} \cdot \underset{\bx}{\E} \left[ \underset{\bg}{\E}\big[\max_{\Lambda \in S} \big|\sum_{j=1}^n f(\Lambda,x_j) g_j\big| \big] \right].
$$}
\state{thm:sym_gauss}
Using the above theorem, we can upper bound 
 the L.H.S. of \eqref{eq:gaussian_ell_1_ball} by 
$$
\sqrt{2 \pi} \cdot \E_{x_1,\ldots,x_m \sim D} \E_{\boldsymbol{g}\sim N(0,1)^m} \big[ \sup_{h \in \FF_1} \big| \sum_{j=1}^m |h(x_j)| \cdot g_j \big| \big] = \sqrt{2 \pi} \cdot \E_{x_1,\ldots,x_m \sim D} \E_{\boldsymbol{g}\sim N(0,1)^m} \big[ \sup_{h \in \FF_1} \big| \big \langle \big( |h(x_i)| \big)_{i \in [m]}, \boldsymbol{g} \big \rangle \big| \big]. 
$$  
In the rest of this section, we bound the right hand side  of the above equation (which is a supremum  of a Gaussian process) using Dudley's entropy integration.
To do this, 
we first extend Lemma 3.7 from \cite{RV08} to bound the covering number of $\FF_1$. 
\begin{definition}
For a set $S \subseteq \mathbb{R}^m$, we define $N(S,\|\cdot\|_2,u)$ (referred to as the covering number) to denote the minimum size of a set $S^0 \subseteq \mathbb{R}^m$ such that $v \in S$, there exists $v^0 \in S^0$ satisfying $\|v-v^0\|_2 \le u$. 
\end{definition}


\begin{claim}\label{clm:covering_ell_1_ball}
Given any $x_1,\ldots,x_m \in D$, $u, \ell \ge 0$,  then the   covering number 
$$
N \left( \left\{ \big( |h(x_i)| \big)_{i \in [m]} \bigg| \|\wh{h}\|_1 \le \ell, \ \supp(\wh{h}) \subseteq T \right\}, \| 
\cdot \|_{2},u \right) \le |T|^{O(m \cdot \log m \cdot \frac{\ell^2}{u^2})}.
$$
\end{claim}
\begin{proof}
For $h,h' \in \FF_1$, the $\ell_2$ distance between the  corresponding vectors $\big( |h(x_i)| \big)_{i \in [m]}$ and $\big( |h'(x_i)| \big)_{i \in [m]}$ can be upper bounded as
$$
\big(\sum_{i \in [m]} \big( |h(x_i)| - |h'(x_i)| \big)^2 \big)^{1/2} \le \big( \sum_{i \in [m]} \big( h(x_i) - h'(x_i) \big)^2 \big)^{1/2} \le \sqrt{m} \cdot \max_{i \in [m]} \big| h(x_i) - h'(x_i) \big|.
$$
This implies that  $$N \big( \big\{ \big( |h(x_i)|  \big)_{i \in [m]} \big| \|\wh{h}\|_1 \le \ell, \ \supp(\wh{h}) \subseteq T  \big\}, \| \cdot \|_{2},u \big) \le N \big(\big\{  \big( h(x_i) \big)_{i \in [m]} \big| \|\wh{h}\|_1 \le \ell, \ \supp(\wh{h}) \subseteq T \big\}, \| \cdot \|_{\infty}, \frac{u}{\sqrt{m}} \big)$$
By rescaling the $\ell_1$ norm of the family, it is enough to prove 
$$
N \big( \big\{ \big( h(x_i)  \big)_{i \in [m]} \big| \|\wh{h}\|_1 \le 1, \ \supp(\wh{h}) \subseteq T\big\}, \| \cdot \|_{\infty},u \big) \le (2|T|)^{\frac{C \log m }{u^2}} \text{ for } C=O(1).
$$
To do this, 
given any $h$ with $\|\wh{h}\|_1=1$, we define a random vector $\bZ \in \mathbb{R}^{|T|}$ as follows -- each coordinate of the vector is indexed by an element in $T$ and we use $\mathbf{e}_{\xi}$ to denote the unit vector which is $1$ in the position corresponding to $\xi$ and $0$ everywhere else. 
Then,  
$$\Pr[\bZ = sign(\wh{h}(\xi)) \cdot \mathbf{e}_{\xi}]  = |\wh{h}(\xi)|.  
$$ 
 Let $\bZ_1,\ldots,\bZ_t$ be i.i.d. copies of $\bZ$ for $t=C \cdot \log m /u^2$ with some large constant $C$.	Observe  that $$\E \bZ = \wh{h} \text{ and consequently } \E \wh{\bZ} = h.$$ 
We now apply Theorem~\ref{thm:sym_gauss} to obtain
\begin{equation}~\label{Zi-Gaussian}
\E_{\bZ_1,\ldots,\bZ_t}\big[ \sup_{j \in [m]}\big| \frac{1}{t} \sum_{i \in [t]} \wh{\bZ_i}(x_j)  - \E\big[ \wh{\bZ}(x_j) \big] \big|\big] \lesssim \frac{1}{t} \cdot \E_{\bZ_1,\ldots,\bZ_t} \E_{\bg} \big[ \sup_{j \in [m]} \big| \sum_{i \in [t]} g_i \cdot \wh{\bZ_i}(x_j) \big| \big], 
\end{equation}
where $\bg = (g_1, \ldots, g_t)$ is a standard $t$-dimensional Gaussian. Next, any point in the support of the random variable $\bZ$ is just a unit vector with $1$ at  one position and $0$ everywhere else; consequently, any point in the support of $\wh{\bZ}$ is a vector whose every coordinate is a complex number of modulus $1$ -- thus, 
for each $j$, $\sum_{i \in [t]} g_i \cdot \wh{\bZ_i}(x_j)$ is a Gaussian random variable with variance $t$.  

Recall that Lemma~\ref{lem:Gaussian_variables} implies that  the maximum of $m$ Gaussian random variables with variance $t$ satisfies $$\E_{g} \big[ \sup_{j \in [m]} \big| \sum_{i \in [t]} g_i \cdot \wh{\bZ_i}(x_j) \big| \big] \lesssim \sqrt{t \cdot \log m}. $$ 
Applying the above to \eqref{Zi-Gaussian}, we obtain 
\begin{equation}\label{eq:sym_gau_Z}
\E_{\bZ_1,\ldots,\bZ_t}\big[ \sup_{j \in [m]}\big| \frac{1}{t} \sum_{i \in [t]} \wh{\bZ_i}(x_j)  - \E\big[ \wh{\bZ}(x_j) \big] \big|\big] \lesssim \sqrt{\frac{\log m}{t}}.
\end{equation} 
This implies that  $\exists z_1,\ldots,z_t \in \supp(\bZ)$ such that $\sup_{j \in [m]}\big| \frac{1}{t} \sum_{i \in [t]} \wh{z_i}(x_j)  - h(x_j) \big| \le u$ for $t=C \cdot \log m /u^2$. Note that $\mathsf{supp}(\bZ) = 2|T|$. Further, $\frac{1}{t} \sum_{i \in [t]} \widehat{z}_i$ (where $\widehat{z}_i \in \mathsf{supp}(\bZ)$)  forms a $u$-cover in $\ell_\infty$ distance for $\{h: \|\wh{h}\|_1 \le 1, \ \supp(\wh{h}) \subseteq T \}$. 
This implies that 
$$
N \big( \big\{ \big( h(x_i)  \big)_{i \in [m]} \big| \|\wh{h}\|_1 \le 1 \big \}, \| \cdot 	\|_{\infty},u \big) \le (2 \cdot |T|)^{t} = (2 \cdot |T|)^{O(\log m /u^2)}.
$$
\end{proof}

\begin{claim}
Given any $x_1,\ldots,x_m$, 
$$
\E_{\boldsymbol{g}} \big[ \sup_{h \in \FF_1} \big| \big \langle \big( |h(x_i)| \big)_{i \in [m]}, \boldsymbol{g} \big \rangle \big| \big] \lesssim \sqrt{m k \log |T| } \cdot \log^{1.5} m.
$$
\end{claim}
\begin{proof}
We apply Dudley's entropy integration \cite{LTbook} to bound the Gaussian process:
$$
\E_{\boldsymbol{g}} \left[ \sup_{h \in \FF_1} \left| \left \langle \big( |h(x_i)| \big)_{i \in [m]}, \boldsymbol{g} \right \rangle \right| \right] \le \int_{0}^{2\sqrt{km}} \sqrt{\log N\left( \left\{ \big( |h(x_j| \big)_{j \in [m]} \big| h \in \FF_1 \right\}, \|\cdot\|_2,u\right)} \mathrm{d} u.
$$ 
For $u$ from $0$ to $1$, we use a covering of size $(\frac{2\sqrt{k \cdot m}}{u} )^{m}$  because $\big( |h(x_j| \big)_{j \in [m]} \in \bigg[ 0,2\sqrt{k} \bigg]^{m}$. Next we use the covering in Claim~\ref{clm:covering_ell_1_ball} to bound the integration from $u=1$ to $2\sqrt{m \cdot k}$:
\begin{align*}
& \int_{0}^{1} \sqrt{\log  (\frac{2\sqrt{k \cdot m}}{u} )^{m} } \mathrm{d} u + 
\int_{1}^{2\sqrt{km}} \sqrt{\log C^{ \log |T| \cdot m \cdot \log m \cdot \frac{k}{u^2} } } \mathrm{d} u\\
 \lesssim & \int_0^{1} \sqrt{m \cdot \log \sqrt{k m} + m \log 1/u} \mathrm{d} u + \int_{1}^{2\sqrt{km}} \sqrt{\log |T| \cdot m \cdot \log m \cdot \frac{k}{u^2} } \mathrm{d} u\\
 \lesssim & \sqrt{m k \log |T|} \cdot \log^{1.5} m.
\end{align*}
\end{proof}
Claim~\ref{clm:ell_1_concentration} follows from the above bound with $m=C \cdot \frac{k \log |T| \cdot \log^3 (k \log |T|/\eps) }{\eps^2}$ for a constant $C$.


%% file: BooleanCube.tex
\section{Sparse FFT over the Boolean cube under random outliers}\label{sec:sparse_FFT_random_outlier}
In this section, we show how to  recover a Fourier-sparse  function over the Boolean cube under random outliers. If the ambient dimension is $n$, we can  instantiate Theorem~\ref{thm:LP_guarantee_boolean_cube} with with the domain $D=\{0,1\}^n$ and $T=\big\{(-1)^{\langle \xi, x \rangle} \big| \xi \in \{0,1\}^n \big\}$ to get an an algorithm with running time $\poly(2^n,1/\delta,1/\eta)$ for $(\rho=1/2-\delta,O(\eta))$ random outlier noise .  The main result of this section is to improve the running time to $\poly(k, n, 1/\delta,1/\eta)$.
\begin{theorem}\label{thm:main_random_boolean_cube}
There is  an algorithm which given as input 
sparsity parameter $k$, input domain $\{0,1\}^n$, parameters $\eta, \delta>0$ and an oracle to $f (x)=\sum_{i=1}^k \wh{f}(\xi_i) (-1)^{\langle \xi_i, x \rangle}$   with  $(\frac{1}{2}-\delta, \eps )$ random outlier noise 
where $|\wh{f}(\xi_i)| \ge \eta$ and $\eps \lesssim \eta \cdot \delta$, with probability $0.99$ outputs $g$ satisfying 
$$
\supp(\wh{g})=\supp(\wh{f}) \text{ and } \E_x[|g(x)-f(x)|] \lesssim \frac{\eps}{\delta}.
$$
The query complexity is $\tilde{O}(k^2 n/\delta^2)$  and  running time  is $\poly(k,n,\frac{1}{\delta \eta})$. 
\end{theorem}
The algorithm in Theorem~\ref{thm:main_random_boolean_cube} is the procedure \textsc{SparseFFTBoolean} (described in Algorithm~\ref{alg:FFT_boolean_random_2}). In this section, we set the function $H:  \fF_2^n \rightarrow \mathbb{R}$ and $\wh{H}:  \fF_2^n \rightarrow \mathbb{R}$ as 
\begin{equation}\label{eq:def_H}
H(x)=2^{n-\ell} \cdot 1_{x_{\ell+1}=\ldots=x_n=0} \text{ and } \quad \wh{H}(\xi)=1_{\xi_1=\xi_2=\ldots=\xi_\ell=0}.
\end{equation}
Note that $\wh{H}$ is the Fourier transform of $H$ over $\fF_2^n$ and $H$ (up to a scaling) is the indicator of a $\ell$-dimensional subspace. The rest of this section is devoted towards proving the correctness of this procedure. 


\begin{algorithm}[H]
\caption{Sparse FFT over Boolean cube of random outliers}\label{alg:FFT_boolean_random_2}
\begin{algorithmic}[1]
\Procedure{\textsc{SparseFFTBoolean}}{$y,k,\delta,\eta$}
\State Let $A \in \fF_2^{n \times n}$ be a random invertible matrix;
\State Set $\ell=2 \log k + 10$ and 
$\mathsf{list}[\xi]=\ast^n$ for each $\xi \in \fF_2^\ell$ \quad  \slash \slash { $\ast^n$ is a string of $\ast$ of length $n$. }
\For{$i \in [n]$}
		\State Sample $b \sim \fF_2^n$ and set $b' = b + e_i$;\quad  \slash \slash { $e_i$ is the standard unit vector in the $i^{th}$ direction. }
		\State Define two oracles $z(x)=y(A x + b) \cdot H(x)$ and $z'(x)=y(A x + b') \cdot H(x)$ for $H$ in \eqref{eq:def_H};
		\State Set $\mathsf{All}_\ell$ to be the set of all characters of $\fF_2^\ell$. 
		\State Apply Procedure~\textsc{LinearDecodingSparseFFT} (i.e.,  Algorithm~\ref{alg:LP_decoding_boolean}) on $(z,\mathsf{All}_\ell,k,\frac{10^{-3}}{k \cdot n},\delta,\eta)$  to obtain $f_z$
		\State Apply Procedure~\textsc{LinearDecodingSparseFFT} (i.e.,  Algorithm~\ref{alg:LP_decoding_boolean}) on  $(z',\mathsf{All}_\ell, k, \frac{10^{-3}}{k \cdot n}, \delta,\eta)$  to obtain $f_{z'}$
		\For{each $\xi \in \supp(\wh{f_z}) \cap \supp(\wh{f_{z'}})$}
			\State $\mathsf{list}[\xi]_i=1_{\sign\big( \wh{f_z}(\xi) \big) \neq \sign\big( \wh{f_{z'}}(\xi) \big)}$ \label{state:indicator_compare}
	\EndFor
\EndFor
\State Set $\supp(\wh{g})=\mathsf{list} \cap \fF_2^n$\\
\Return $g=\underset{h \in \mathsf{span}(\wh{g})}{\arg\min} \sum_{i=1}^m |h(x_i)-y(x_i)|$ with $m=\wt{O}(k^2 n/\delta^2)$ random points $x_1,\ldots,x_m$.
\EndProcedure
\end{algorithmic}
\end{algorithm}

We begin with a few  useful definitions and technical lemmas.  
Given any function $f:\fF_2^n \rightarrow \mathbb{R}$ with $\supp(\wh{f})$ of size at most $k$, a matrix $A \in \fF_2^{n \times n}$, we define the notion of \emph{isolated frequencies}. 
\begin{definition}~\label{def:isolated-frequencies}
Given a matrix $A$ and $f(x)=\sum_{j=1}^k \wh{f}(\xi_j) \cdot (-1)^{ \langle \xi_j, x \rangle}$, we say $A$ isolates a frequency $\xi_j$ in $f$ only if $\forall j' \neq j, (A^\top \xi_{j'})_{[\ell]} \neq (A^\top \xi_{j})_{[\ell]}$.
\end{definition}
Next we show \emph{all} frequencies are isolated with high probability.
\begin{claim}\label{clm:isolated_A}
Given any $k$ frequencies $\xi_1,\ldots,\xi_k$, for $\ell \ge 2 \log k + 10$ and a random invertible matrix $A \in \fF_2^n$, all $\xi_i$ are isolated with probability $1-10^{-3}$.
\end{claim}
\begin{proof}
For any $\xi_{i'} \neq \xi_i$, $\underset{A}{\Pr}\big[(A^\top \xi_i)_{[\ell]} = (A^\top \xi_{i'})_{[\ell]} \big] \le 2^{-\ell}$ because 
$$
\underset{A \sim \fF_2^{n \times n}}{\Pr}\left[ (A^\top \xi_i)_{[\ell]} = (A^\top \xi_{i'})_{[\ell]} \right] = \underset{A \sim \fF_2^{n \times n}}{\Pr}\left[ \big( A^\top (\xi_i - \xi_{i'}) \big)_{[\ell]} \right] = 2^{-\ell}
$$ and this probability only increases conditioned on the event that $A$ is singular. 

Over a union bound for all pairs $i, i' \in [k]$, $\xi_i$ is isolated with probability $1-{k \choose 2} 2^{-\ell} \ge 1-10^{-3}$.
\end{proof}

Given $A$ and $b \in \fF_2^n$, we define $g_{A,b}$ to be the shift $g_{A,b}(x)=g(Ax+b)$.
\begin{claim}~\label{clm:Fourier-linear-transform}
$\wh{g_{A,b}}(\xi)=(-1)^{ \langle b, (A^{\top})^{-1} \xi \rangle} \cdot \wh{g}\big( (A^{\top})^{-1} \xi \big).$
\end{claim}
\begin{proof}
From the definition, 
\begin{align*}
\wh{g_{A,b}}(\xi)&=\E_{x} \left[(-1)^{\langle x, \xi \rangle} \cdot g_{A,b}(x) \right]
\\
	&=\E_{x}\left[ (-1)^{\langle x, \xi \rangle} g(Ax+b) \right]
\\
	&=\E_{x}\left[ (-1)^{ \langle b, (A^{\top})^{-1} \xi \rangle} \cdot (-1)^{ \langle A x + b, (A^{\top})^{-1} \xi \rangle} g(Ax+b) \right] 
\\
	&= (-1)^{ \langle b, (A^{\top})^{-1} \xi \rangle} \cdot \wh{g}\big( (A^{\top})^{-1} \xi \big).
\end{align*}
\end{proof}
Next, for any $d$, $\ell \le d$ and any vector $v$ of dimension $d$, we use $v_{[\ell]}$ to denote the first $\ell$ coordinates of $v$.  We next have the following claim. 
\begin{claim}\label{clm:good_isolation}
For $g(x)=\sum_{j=1}^k \wh{g}(\xi_j) (-1)^{\langle \xi_j,x \rangle}$ and $z(x)=g(A x + b) \cdot H(x)$, $$\wh{z}(\xi)=\underset{\xi_j:(A^\top \xi_j)_{[\ell]}=\xi_{[\ell]}}{\sum} (-1)^{ \langle b, \xi_j \rangle} \cdot \wh{g}(\xi_j).$$
\end{claim} 
\begin{proof}
From the definition, 
$$
\wh{z}(\xi)=(\wh{g_{A,b}} * \wh{H})(\xi)=\sum_{\xi'} \wh{g_{A,b}}(\xi') \wh{H}(\xi-\xi').
$$
Since $\wh{H}(\xi-\xi')=1$ iff $\xi_1=\xi'_1, \ldots, $ and $\xi_\ell=\xi'_\ell$, the above expression simplifies to (using Claim~\ref{clm:Fourier-linear-transform}) 
$$
\wh{z}(\xi)=\sum_{\xi': \xi'_1=\xi_1,\ldots,\xi'_\ell=\xi_\ell} \wh{g_{A,b}}(\xi') = \sum_{\xi': \xi'_1=\xi_1,\ldots,\xi'_\ell=\xi_\ell} (-1)^{ \langle b, (A^{\top})^{-1} \xi' \rangle} \cdot \wh{g}\big( (A^{\top})^{-1} \xi' \big).
$$
Now observe that $\wh{g}(\zeta) \not =0$ iff $\zeta = \xi_j$ for $1 \le j \le k$. With this substitution, we get 
$$
\wh{z}(\xi)=\sum_{\xi_j \in \fF_2^n:(A^\top \xi_j)_{[\ell]}=\xi_{[\ell]}} (-1)^{ \langle b, \xi_j \rangle} \cdot \wh{g}(\xi_j).
$$
\end{proof}




 We state a direct corollary of Claim~\ref{clm:good_isolation} for isolated frequencies.
\begin{corollary}\label{cor:isolated_freq_A}
Given $f (x)=\sum_{i=1}^k \wh{f}(\xi_i) (-1)^{\langle \xi_i, x \rangle}$ and a non-singular matrix $A$ where $\xi_j$ is isolated by $A$, for $z=f(Ax+b) \cdot H(x)$, $\wh{z}\big( (A^\top \xi_{j})_{[\ell]} \big)=(-1)^{\langle b, \xi_j \rangle} \cdot \wh{f}(\xi_j)$.
\end{corollary}
We now argue the following guarantee for the procedure~\textsc{SparseFFTBoolean} (Algorithm~\ref{alg:FFT_boolean_random_2}). 

\begin{claim}\label{clm:recover_each_coordinate}
Given  $\eta, \delta,$ and $f(x)=\sum_{j=1}^k \wh{f}(\xi_j) \cdot (-1)^{ \langle \xi_j, x \rangle}$ with $\wh{f}(\xi_j) \ge \eta$, let $A$ be a non-singular matrix such that all $\xi_j$ in $f$ (where $1 \leq j \leq k$) are isolated by $A$. 

If $y(\cdot)$ is an oracle for $f$ with $( \frac{1}{2}-\delta, \eps)$ random outlier noise with $\eps \lesssim \eta \cdot \delta$, for any  frequency $\xi_j \in \supp(\wh{f})$ and coordinate $i$, the indicator function $1_{\sign\big( \wh{f_z}(\xi) \big) \neq \sign\big( \wh{f_{z'}}(\xi) \big)}$ of $\xi=(A^\top \xi_j)_{[\ell]}$ in Line~\ref{state:indicator_compare} of Procedure~\textsc{SparseFFTBoolean} equals the $i$th bit of $\xi_j$ with probability $1-\frac{2 \cdot 10^{-3}}{k \cdot n}$.
\end{claim}
\begin{proof}
Let us begin with the \emph{noiseless} case. Here, by Corollary~\ref{cor:isolated_freq_A}, 
$$
\wh{z}\bigg( (A^\top \xi_{j})_{[\ell]} \bigg)=(-1)^{\langle b, \xi_j \rangle}\wh{f}(\xi_j) \text{ and } \wh{z'}\bigg( (A^\top \xi_{j})_{[\ell]} \bigg)=(-1)^{\langle b+e_i, \xi_j \rangle} \wh{f}(\xi_j).
$$
By comparing $\sign\big( \wh{z}(\xi) \big)$ and $\sign\big( \wh{z'}(\xi) \big)$ of $\xi=(A^\top \xi_j)_{[\ell]}$, we could decode the $i$th bit of $\xi_j$ through $(-1)^{\langle e_i,\xi_j \rangle}$.

Under random outliers, we use the guarantee of Theorem~\ref{thm:LP_guarantee_boolean_cube}. From the assumption of $A$, all frequencies in $f$ are isolated. Thus $z$ and $z'$ are oracles of $k$-Fourier-sparse functions with each coefficient at least $\eta$. The guarantee of Theorem~\ref{thm:LP_guarantee_boolean_cube} on $z$ shows, with probability $1-\frac{10^{-3}}{k \cdot n}$, 
$$
| (-1)^{\langle b, \xi_j \rangle} \cdot \wh{f}(\xi_j)   - \wh{f_z}(\xi)| \lesssim \frac{\eps}{\delta} \le 0.1 \eta.
$$ 

By the same argument, 
$$
| (-1)^{\langle b + e_i, \xi_j \rangle} \cdot \wh{f}(\xi_j)   - \wh{f_{z'}}(\xi)| \lesssim \frac{\eps}{\delta} \le 0.1 \eta.
$$ 
These two imply that with prob $1 - \frac{2 \cdot 10^{-3}}{k \cdot n}$, $\sign\big( \wh{f_z}(\xi) \big) \neq \sign\big( \wh{f_{z'}}(\xi) \big)$ when the $i$th bit of $\xi_j$ is 1. On the other hand, $\sign\big( \wh{f_z}(\xi) \big) = \sign\big( \wh{f_{z'}}(\xi) \big)$ when the $i$th bit of $\xi_j$ is 0. 
\end{proof}

Finally, we prove the correctness of Procedure~\textsc{SparseFFTBoolean}.

\begin{proofof}{Theorem~\ref{thm:main_random_boolean_cube}}
By Claim~\ref{clm:isolated_A}, all frequencies in $f$ are isolated by a random non-singular matrix with probability $1-10^{-3}$. Then for each frequency $\xi_j$ and each coordinate $i$, Claim~\ref{clm:recover_each_coordinate} shows that $1_{\sign\big( \wh{f_z}(\xi) \big) \neq \sign\big( \wh{f_{z'}}(\xi) \big)}$ of $\xi=(A^\top \xi_j)_{[\ell]}$ equals $\xi_j[i]$ with probability $1-\frac{2 \cdot 10^{-3}}{k \cdot n}$. By a union bound over all $\xi_j$ and $i$, with probability $1-2 \cdot 10^{-3}$, $\mathsf{list}\bigg[ (A^\top \xi_j)_{[\ell]} \bigg]=\xi_j$ for all $\xi_j$. 

On the other hand, there are at most $k$ elements of $\mathsf{list}$ in $\fF_2^{n}$ because $\mathsf{list} \subseteq \supp(\wh{f_z}) \cap \supp(\wh{f_{z'}})$ and Theorem~\ref{thm:LP_guarantee_boolean_cube} only returns $k$-Fourier-sparse functions (i.e., $\supp(\wh{g})=\supp(\wh{f})$ in the guarantee). Thus $\supp(\wh{g})=\supp(\wh{f})$ and the distance $\E[|f-g|]$ follows from Lemma~\ref{cor:correctness_over_k_sparse}.

Finally, we bound the running time and query complexity of our algorithm by $2n$ times the counterparts of Procedure \textsc{LinearDecodingSparseFFT} over the domain $\{0,1\}^{O(\log k)}$, which are $\poly(n, k, 1/\delta, 1/\eta)$ and $\wt{O}(k^2 n/\delta^2)$ separately.
\end{proofof}

%% file: torus.tex
\section{Sparse FFT of periodic signals under random outliers}\label{sec:periodic_sparse_FFT_rand}
In this section, we consider the recovery of a periodic signal $f:[0,1) \rightarrow \mathbb{R}$ with $\|\wh{f}\|_0 \le k$ under the random outlier noise. Given the bandlimit $F$ of frequencies, Theorem~\ref{thm:LP_guarantee_boolean_cube} implies an algorithm with running time $\poly(F,1/\delta,1/\eta)$ under the $(\rho=1/2-\delta,O(\eta))$ random outlier noise with the domain $D=[0,1)$ and the set $T=\bigg\{e^{2 \pi \bi \cdot \xi t} \big| \xi \in [-F,F] \cap \mathbb{Z}\bigg\}$. Our main result is to improve the running time to $\poly(k, \log F,1/\delta,1/\eta)$.

\begin{theorem}\label{thm:periodic_FFT_rand_out}
Given the sparsity $k$, the band limit $F$, any $\delta>0$, and $\eta>0$, there exists an algorithm with running time $\poly(k, \log F, 1/\delta,1/\eta)$ and {$\tilde{O}(k^2 \log F/\delta^2)$ queries} such that for any $f(t)=\sum_{j=1}^k \wh{f}(\xi_j) \cdot e^{2 \pi \bi \xi_j \cdot t}$ with each $\xi_j \in [-F,F]$ and $|\wh{f}(\xi_j)| \ge \eta$, under the $( \frac{1}{2}-\delta, \eps)$ random outlier noise of $\eps \lesssim \eta \cdot \delta$, it outputs $g$ satisfying 
$$
\supp(\wh{g})=\supp(\wh{f}) \text{ and } \E_x[|g(x)-f(x)|] \lesssim \frac{\eps}{\delta}, \textit{ with probability } 0.99.
$$
{Further, with high probability, the query points of the algorithm are at least $1/\poly(k, \log F)$ apart from each other. }
\end{theorem}

We describe our algorithm in Algorithms~\ref{alg:torusFFT_randomoutliers_2} and~\ref{alg:torusFFT_randomoutliers_1}. We will use the following two notations: For any complex number $z = r \cdot e^{i\theta}$ (where $r \ge 0$ and $\theta \in [-\pi, \pi)$), we use $\Phi(z) = \theta$ to denote its phase. Also, for any real $x$, we use $\mathsf{round}(x)$ to denote the nearest integer to $x$. 
\begin{algorithm}[H]
\caption{Sparse FFT for periodic signals under random outliers}\label{alg:torusFFT_randomoutliers_2}
\begin{algorithmic}[1]
\Procedure{\textsc{PeriodicSparseFFTunderRandomOutliers}}{$y,F,\delta,\eta$}
\State Let $P$ be a subset of primes greater than $(k \log F/\delta)^{10}$ of size at least {$10^3 k^2 \log F$}
\State Sample $B \sim P$
\State Set $\Delta=1/4F$ 
\State Set $\midd [\xi]=0$ for each $\xi \in [B]$
\For{$i \in [\log_2 2F]$}
	\State\label{st:start_loop_torus} Apply Procedure\textsc{FrequencyHash} with $(y,B,\Delta,\delta,\eta)$ to obtain $f_z$ and $f_{z'}$
	\For{each $\xi \in [B]$}
		\If{$\midd[\xi] \neq \nulll$ and $\xi \in \supp(\wh{f_z})
	 \cap \supp(\wh{f_{z'}})$}
			\State Set $\gamma = \Phi\left( e^{-2 \pi \bi \Delta \cdot \midd[\xi]} \cdot  \wh{f_{z'}}(\xi)/\wh{f_z}(\xi) \right) \in [-\pi, \pi)$
			\State update $\midd[\xi]=\midd[\xi]+ \mathsf{round}\left(\frac{\gamma}{2 \pi \cdot \Delta} \right)$
		\Else
			\State set $\midd[\xi]=\nulll$
		\EndIf
	\EndFor
	\State update $\Delta=2\Delta$
\EndFor
\State Set $\supp(\wh{g})=\bigg\{ e^{2\pi \bi x \cdot \xi} \bigg| \xi \in \midd \cap [-F,F]\bigg \}$\\
\Return $g=\underset{h \in \mathsf{span}(\wh{g})}{\arg\min} \sum_{i=1}^m |h(x_i)-y(x_i)|$ with $m=\wt{O}(k^2/\delta^2)$ random points $x_1,\ldots,x_m$.

\EndProcedure
\end{algorithmic}
\end{algorithm}

\begin{algorithm}[h]
\caption{Sparse FFT for periodic signals under random outliers}\label{alg:torusFFT_randomoutliers_1}
\begin{algorithmic}[1]
\Procedure{\textsc{FrequencyHash}}{$y,B,\Delta,\delta,\eta$}
\State Sample $t_0 \sim [0,1/B)$
\State Set an oracle $z$ over $\mathbb{C}^B$ as $z[i]=y\big( t_0+(i-1)/B \big)$.
\State Set an oracle $z'$ over $\mathbb{C}^B$ as $z'[i]=y\big( t_0+\Delta+(i-1)/B \big)$.
\State Set $\mathsf{All}_B$ to be the set of all characters in $\mathbb{Z}_B$. 
\State Apply Procedure~\textsc{LinearDecodingSparseFFT}  (Algorithm~\ref{alg:LP_decoding_boolean}) with $(z,\mathsf{All}_B,k,\frac{10^{-3}}{k \cdot \log_2 2F},\delta,\eta)$  to obtain $f_z$
\State Apply Procedure~\textsc{LinearDecodingSparseFFT} in (Algorithm~\ref{alg:LP_decoding_boolean}) with $(z',\mathsf{All}_B,k,\frac{10^{-3}}{k \cdot \log_2 2F}, \delta,\eta)$  to obtain $f_{z'}$
\\ \Return $f_z$ and $f_{z'}$
\EndProcedure
\end{algorithmic}
\end{algorithm}

In the rest of this section, we prove the correctness of Procedure~\textsc{PeriodicSparseFFTunderRandomOutliers} in Algorithm~\ref{alg:torusFFT_randomoutliers_2} for Theorem~\ref{thm:periodic_FFT_rand_out}. 

Given $f: [0,1) \rightarrow \mathbb{R}$ defined as $f(t)=\sum_{j=1}^k \wh{f}(\xi_j) \cdot e^{2 \pi \bi \xi_j \cdot t}$ with $k$ frequencies $\xi_1,\ldots,\xi_k$, we say $B$ isolates a frequency $\xi_j$ iff $\forall j' \in [k]\setminus\{j\}, \xi_{j'} \not\equiv \xi_j \mod B$.  We first show that a random $B$ will isolate \emph{all} frequencies in any such $k$-Fourier sparse function $g$ with high probability.
\begin{claim}\label{clm:isolated_B}
Let $P$ be a subset of primes greater than $2 k \log F$ of size at least $10^3 (k \log F)^2$. For any $g(x)=\sum_{j=1}^k v_j e^{2 \pi \bi \xi_j x}$, for a random prime $B \sim P$, all $\xi_j$ will be isolated with probability at least $1-10^{-3}$.
\end{claim}
\begin{proof}
Observe that $\xi_j$ is not isolated only if $$\prod_{j' \neq j}|\xi_{j'}-\xi_j| \equiv 0 \mod B.$$

Since the product is at most $(2F)^{k-1}$, the number of its primes factors greater than $2 k \log F$ is at most $(k-1) \frac{\log 2F}{ \log (2 k \log F)}$. This shows each $\xi_j$ is isolated with probability at least $1-\frac{(k-1) \log F}{|P|}$. Plugging the value of $|P|$, we get that  all frequencies are isolated with probability at least $1-10^{-3}$.
\end{proof}
\begin{fact}\label{fact:hashtobins}
Given $B$, $\Delta$, and $f(t)=\sum_{j=1}^k \wh{f}(\xi_j) e^{2 \pi \bi \xi_j t}$ with $k$ integer frequencies $\xi_1,\ldots,\xi_k \in [-F,F]$, and $t_0 \in [0,1/B)$, let us define $z,z':\mathbb{Z}_B \rightarrow \mathbb{C}$ as follows:
 \[z[i]=f\big( t_0+(i-1)/B \big) \ \ \text{and} \ \ z'[i]=f\big( t_0+\Delta+(i-1)/B \big). 
 \] Then 
$$
\wh{z}(\ell)=\underset{j \in [k]: \xi_j \equiv \ell \mod B}{\sum} \wh{f}(\xi_j) \cdot e^{2 \pi \bi \xi_j \cdot t_0} \text{ and } \wh{z'}(\ell)=\underset{j \in [k]: \xi_j \equiv \ell \mod B}{\sum} \wh{f}(\xi_j) \cdot e^{2 \pi \bi \xi_j \cdot (t_0+\Delta)}.
$$
\end{fact}

\begin{proof}
It is enough to consider 
$$
z'[i]=f\big( t_0 + \Delta+(i-1)/B \big)=\sum_{j} \wh{f}(\xi_j) \cdot e^{2 \pi \bi \xi_j \cdot (t_0+\Delta)} \cdot e^{2 \pi \bi \xi_j \cdot (i-1)/B}.
$$
This immediately gives that the Fourier transform of $z'$ is  $\wh{z'}(\ell)=\underset{j: \xi_j \equiv \ell \mod B}{\sum} \wh{f}(\xi_j) \cdot e^{2 \pi \bi \xi_j \cdot (t_0+\Delta)}$.
\end{proof}
A direct corollary of Fact~\ref{fact:hashtobins} is that when $\xi_j$ is isolated under $B$ and $z$ and $z'$ are defined as in Fact~\ref{fact:hashtobins}, then
\[
\frac{\wh{z'}(\xi_j \ \mathsf{mod}\ B)}{\wh{z}(\xi_j \ \mathsf{mod} \ B)}=e^{2 \pi \bi \xi_j \Delta}. 
\]
When $\Delta < 1/2F$, then applying the phase function, we have 
\[
\Phi\bigg(\frac{\wh{z'}(\xi_j \ \mathsf{mod}\ B)}{\wh{z}(\xi_j \ \mathsf{mod} \ B)}  \bigg) =  2 \pi \xi_j \Delta \in (-\pi,\pi). 
\]
The intuition behind why Claim~\ref{clm:isolated_B} and Fact~\ref{fact:hashtobins} are useful is as follows: Consider the case when the oracle to $f$ (call it $y(\cdot)$) is \emph{noiseless}. We can choose a prime $B$ as done in Step~2 and 3 of Algorithm~\ref{alg:torusFFT_randomoutliers_1} -- Claim~\ref{clm:isolated_B} says that with high probability, all the frequencies of $f$ are isolated. Given the oracle $y(\cdot)$, we can compute both $z$ and $z'$ and thus use Fact~\ref{fact:hashtobins} to get all the frequencies appearing in the spectrum of $f$. 

We now state another lemma 
(the main technical lemma concerning Procedure~\textsc{PeriodicSparseFFTunderRandomOutliers})
. We defer the proof of this lemma to Section~\ref{sec:proof_recover_isolated}. 


\begin{lemma}\label{lem:recover_isolated_freq}
Given $f(t)=\sum_{j=1}^k \wh{f}(\xi_j) \cdot e^{2 \pi \bi \xi_j \cdot t}$ with each $\xi_j \in [-F,F]$ and $|\wh{f}(\xi_j)| \ge \eta$, let $B$ be a prime number (selected in Step~3 of  Procedure \textsc{PeriodicSparseFFTunderRandomOutliers}) which isolates  all frequencies $\xi_j$ in $f$. If $y(\cdot)$ is an oracle for $f$ with  $( \frac{1}{2}-\delta, \eps)$ random outlier noise  (where $\eps \lesssim \eta \cdot \delta$), then, with probability $1-5 \cdot 10^{-3}$, after the for loop of $i \in [\log_2 (2F)]$ in Procedure \textsc{PeriodicSparseFFTunderRandomOutliers}, $\midd \cap [-F,F]$ is the support set of $f$, i.e., $\{\xi_1, \ldots, \xi_k\}$.
\end{lemma}
We are now ready to  prove the correctness of ~\textsc{PeriodicSparseFFTunderRandomOutliers}.

\begin{proofof}{Theorem~\ref{thm:periodic_FFT_rand_out}}
First of all, observe that time (resp. sample) complexity of Steps 1 to 18 is $\log (2F)$ times the time (resp. sample) complexity of a single iteration (defined in Step~6).  The sample complexity of each iteration is $\wt{O}(k^2 \log B/\delta^2)$ from Theorem~\ref{thm:LP_guarantee_boolean_cube} and the time complexity is $\poly(B)$. 
Finally, the sample  complexity of Step 19 is $\tilde{O}(k^2/\delta^2)$ and time complexity is $\poly(k,1/\delta)$. This gives the claimed bounds on the time and sample complexity.

Next, 
Claim~\ref{clm:isolated_B} shows a random $B$ will isolate all frequencies with probability $1-10^{-3}$. Lemma~\ref{lem:recover_isolated_freq} then shows that with probability $1-5 \cdot 10^{-3}$, $\supp(\wh{g})=\supp(\wh{f})$. The guarantee of $g$ follows from the correctness of Lemma~\ref{cor:correctness_over_k_sparse}. 

Finally we notice that Algorithm~\ref{alg:LP_decoding_boolean} only queries $m=\wt{O}(k^2 \log B /\delta^2)$ points in each invocation of Procedure \textsc{FrequencyHash}. For any $\Delta$ and $t_0$, the probability that two query points of our algorithm are closer than $1/B$ is at most $2/B$. At the same time, because the total number of invocations of Algorithm~\ref{alg:LP_decoding_boolean} is $2 \log 2F$, the total number of query points is $m \cdot (2 \log 2F)$. So the probability that any two query points are $1/B$ close is at most ${m \cdot (2 \log 2F) \choose 2} \cdot 2/B \le 1/k$. Overall, this implies that all the query points, with high probability are at least $1/\poly(k, \log F)$ far from each other.
\end{proofof}

\subsection{Proof of Lemma~\ref{lem:recover_isolated_freq}}\label{sec:proof_recover_isolated}
First of all, by Claim~\ref{clm:isolated_B}, all frequencies $\{\xi_\ell\}_{\ell=1}^k$ are isolated under 
$B$ with probability at least $1-10^{-3}$. We now condition on this event -- i.e., all frequencies $\{\xi_\ell\}_{\ell=1}^k$ are isolated under 
$B$. 

Given all frequencies $\xi_\ell$ isolated under $B$, we use induction to prove that with high probability, all frequencies $\xi_\ell$ satisfy the following 
\begin{equation}~\label{eq:induction}
\xi_\ell \in \left[\midd[\xi_\ell \mod B] - \frac{1}{4\Delta},\midd[\xi_\ell \mod B] + \frac{1}{4\Delta} \right],
\end{equation} 
for every $i$ in the loop (defined in Step~6) 
of procedure \textsc{PeriodicSparseFFTunderRandomOutliers} 
\ignore{of Procedure \textsc{FreqEstimation} in Algorithm~\ref{alg:torusFFT_randomoutliers_1}.} 
Observe that $\midd[\xi_\ell \ \mathsf{mod} \ B]$ (unless, it is $\nulll$), is always an integer. Now, at the end of the loop, $\Delta= 1/2$.  Thus, applying (\ref{eq:accuracy}), at the end of the loop, for all $\xi_\ell$ such that $\midd[\xi_\ell \mod B]$ is not null, we have 
\[
\xi_\ell \in \left[\midd[\xi_\ell \mod B] - \frac{1}{2},\midd[\xi_\ell \mod B] + \frac{1}{2} \right]. 
\]
But this immediately implies that at the end of the loop that $\xi_\ell = \midd[\xi_\ell \ \mathsf{mod} \ B]$. Thus, it just remains to prove \eqref{eq:induction}.

\textbf{Base case:} At the beginning of the iteration, $\midd[\chi_\ell \ \mathsf{mod} \ B] =0$ and $\Delta = 1/(4F)$. Since all $\xi_\ell \in [-F, F]$ (by assumption), we have 
\[
\xi_\ell \in \left[\midd[\xi_\ell \mod B] -  \frac{1}{4\Delta},\midd[\xi_\ell \mod B] + \frac{1}{4\Delta} \right],  
\]
at the beginning of the loop. 

\textbf{Induction step:} Suppose for $i =i_0 \in [\log 2F]$, it holds that 
\[
\xi_\ell \in \left[ \midd[\xi_\ell \ \mathsf{mod} \  B] -  \frac{1}{4\Delta}, \midd[\xi_\ell \ \mathsf{mod} \  B] +  \frac{1}{4\Delta} \right]. 
\]
We will show that this relation holds for $i = i_0 +1$ as well. By assumption, all frequencies of $f$ are isolated by $B$. Since $|\widehat{f}(\xi_\ell) | \ge \eta$, by Theorem~\ref{thm:LP_guarantee_boolean_cube} implies that with probability $1-\frac{10^{-3}}{k \log 2|F|}$, the output of Procedure~\textsc{LinearDecodingSparseFFT} (invoked in Line 6 of \textsc{FrequencyHash}), the function $f_z$ satisfies 
\[
\bigg| \wh{f}(\xi_\ell) \cdot e^{2 \pi \bi \xi_j \cdot t_0} - \wh{f_z}(\xi_\ell \mod B) \bigg| \lesssim \frac{\eps}{\delta} \le 0.1 \eta \text{ for every } \xi_\ell.
\]
By the same argument, 
\[
\bigg| \wh{f}(\xi_\ell) \cdot e^{2 \pi \bi \xi_j \cdot (t_0+\Delta)} - \wh{f_{z'}}(\xi_\ell \mod B) \bigg| \lesssim \frac{\eps}{\delta} \le 0.1 \eta \text{ for every } \xi_\ell.
\]
Together, the above equations (using $|\wh{f}(\xi_\ell)| \ge \eta$) imply 
\begin{equation}\label{eq:accuracy}
\bigg| \wh{f_{z'}}(\xi_\ell \mod B)/\wh{f_z}(\xi_\ell \mod B) - e^{2 \pi \bi \xi_\ell \cdot \Delta} \bigg| \le 0.35 
\end{equation}
Multiplying by $e^{-2 \pi \bi \Delta \cdot \midd[\xi_\ell \ \mathsf{mod} \ B]}$, 
\begin{equation}~\label{eq:closeness-angle}
\bigg|  e^{-2 \pi \bi \Delta \cdot \midd[\xi_\ell \ \mathsf{mod} \ B]} \cdot \wh{f_{z'}}(\xi_\ell \mod B)/\wh{f_z}(\xi_\ell \mod B) - e^{-2 \pi \bi \Delta \cdot \midd[\xi_\ell \ \mathsf{mod} \ B]} \cdot e^{2 \pi \bi \xi_\ell \cdot \Delta} \bigg| \le 0.35 
\end{equation}
Now, observe that  by induction hypothesis, for every $\xi_\ell$, 
$$
\Delta \cdot (\xi_\ell - \midd[\xi_\ell \ \mathsf{mod} B])  \in \bigg[ -\frac14, \frac14\bigg]. 
$$
Combining with \eqref{eq:closeness-angle}, 
\begin{equation}~\label{eq:closeness-angle2}
\bigg|  \Phi\bigg(e^{-2 \pi \bi \Delta \cdot \midd[\xi_\ell \ \mathsf{mod} \ B]} \cdot \wh{f_{z'}}(\xi_\ell \mod B)/\wh{f_z}(\xi_\ell \mod B) \bigg)  
- 2 \pi \Delta \bigg( \xi_\ell - \midd[\xi_\ell \ \mathsf{mod} \ B] \bigg) 
\bigg| \le \frac{\pi}{4}. 
\end{equation}
Note that $\gamma$, defined in  Step~10 of the algorithm is precisely 
\[
\gamma=
\Phi\bigg(e^{-2 \pi \bi \Delta \cdot \midd[\xi_\ell \ \mathsf{mod} \ B]} \cdot \wh{f_{z'}}(\xi_\ell \mod B)/\wh{f_z}(\xi_\ell \mod B) \bigg) . 
\]
Then, it follows  from \eqref{eq:closeness-angle2}, 
\[
\mathsf{round}\bigg( \frac{\gamma}{2\pi \cdot \Delta}\bigg) \in \left[ \xi_\ell-  \midd[\xi_\ell \ \mathsf{mod} \ B] - \frac{1}{8\Delta}, \xi_\ell-  \midd[\xi_\ell \ \mathsf{mod} \ B] +  \frac{1}{8\Delta} \right]
\]
However, at the end of the $i^{th}$ round, $\Delta$ gets updated to $2\Delta$. Thus, the induction hypothesis continues to hold for $(i+1)^{th}$ iteration. 

\ignore{
Next we use \eqref{eq:accuracy} to approximate $e^{2 \pi \bi \xi_l \cdot \Delta}$ in the equation $\phi\left( e^{-2 \pi \bi \Delta \cdot mid[\xi_l \mod B]} \cdot  e^{2 \pi \bi \xi_l \cdot \Delta} \right) = 2 \pi \Delta \cdot (\xi_l - mid[\xi_l \mod B])$ such that under the induction hypothesis $\Delta \cdot (\xi_l - mid[\xi_l \mod B]) \in [-1/4,1/4]$,
$$
\phi\left( e^{-2 \pi \bi \Delta \cdot mid[\xi_l \mod B]} \cdot  \wh{z'}(\xi_l \mod B)/\wh{z}(\xi_l \mod B) \right) = 2 \pi \Delta \cdot (\xi_l - mid[\xi_l \mod B]) \pm \pi/4.
$$
Thus the algorithm updates $mid[\xi_l \mod B]$ by $(\xi_l - mid[\xi_l \mod B]) \pm \frac{1}{8\Delta}$, which satisfies the induction hypothesis on $\Delta'$. By a union bound, this is true for every frequency $\xi_l$ and $i$ with probability $1 - 2 \cdot 10^{-3}$.}

%% file: adv_outliers_boolean.tex
\section{Low Degree functions under adversarial outliers}\label{sec:adv_outlier_boolean}
In this section, we give an efficient algorithm 
for Fourier sparse functions which can tolerate $\rho = 
\frac{1}{4 \cdot 3^{2d}}$ of adversarial outliers when the target function has degree-at most $d$. Recall that arbitrary $k$-Fourier-sparse signals can tolerate up to $\Theta(1/k)$ fraction of adversarial outliers. Since degree-$d$ polynomials are $\binom{n}{\le d}$-Fourier sparse over $\{0,1\}^n$ (in the worst case), this translates to an error tolerance of  $(
{n \choose \le d})^{-1}$ fraction of adversarial outliers. By exploiting the low-degree structure of the target function, we are able to improve this error tolerance to $(4 \cdot 3^{2d})^{-1}$.



\begin{theorem}\label{thm:degree_d_adv_outliers}
There is an efficient algorithm (Algorithm~\ref{alg:LP_decoding_linear_GZ}) which when given as input parameters $\delta>0$, $\eps>0$, and $\rho \le \frac{1}{4 \cdot 3^{2d}}-\delta$ and an oracle $y (\cdot)$ to a degree-$d$ polynomial $f$ over $\{0,1\}^n$ with $(\rho,\eps)$ adversarial outlier noise, makes $m$ random queries and outputs a function $g$ which with probability $0.99$ satisfies 
$$ |\wh{g}(\xi)-\wh{f}(\xi)| \lesssim \frac{\eps}{\delta} \text{ for each $\xi$ and } \E_x[|g(x)-f(x)|^2]^{1/2} \lesssim \frac{3^{d} \cdot \eps}{\delta}.$$
Here the query complexity $m = O\bigg( \frac{{n \choose \le d} \log {n \choose \le d}}{\delta^2} \bigg)$. 
\end{theorem}

\begin{algorithm}[H]
\caption{Recover degree $d$ functions}\label{alg:LP_decoding_linear_GZ}
\begin{algorithmic}[1]
\Procedure{\textsc{RecoverLowDegree}}{$y,\rho,\epsilon,n,d$}
\State Set $m=O\bigg( \frac{{n \choose \le d }\log {n \choose \le d} }{\delta^2} \bigg)$;
\State Sample $x_1,\ldots,x_m$ randomly;
\State Query $y(x_1),\ldots,y(x_m)$;
\State Find the degree $d$ function $g(x)$ minimizing $\sum_{i=1}^m |g(x_i)-y(x_i)|$;
\EndProcedure
\end{algorithmic}
\end{algorithm} 

In the rest of this section, we finish the proof of Theorem~\ref{thm:degree_d_adv_outliers}. We first observe that algorithm can be implemented efficiently  i.~e.~, in time $\poly(m,n)$. The key fact is that Step~5 can be implemented as a linear program (with the unknowns being the $\binom{n}{ \le d}$ coefficients of the polynomial $g(x)$). Thus, what remains to be done is to establish the performance of this algorithm. 
 The key here  is to obtain  concentration bounds  on $\sum_{i=1}^m |p(x_i)|^2$ and $\sum_{i=1}^m |p(x_i)|$ for all degree $d$ functions $p$. 
We first state our concentration bounds of $\sum_{i=1}^m |p(x_i)|^2$.
\begin{theorem}\label{thm:concentration_deg_d_poly_ell2}
There exists a constant $C>0$ such that for 
 any $\epsilon$ and $\delta$ and  $m = C \frac{{n \choose \le d} \log \frac{{n \choose \le d}}{\delta}}{\epsilon^2}$, if  $x_1,\ldots,x_m \sim \{0, 1\}^n$, then  with probability $1-\delta$, for any degree $d$ function $p$,
\begin{equation}\label{eq:ell2norm_deg_d_poly}
\sum_{i=1}^m p(x_i)^2 \in [1-\epsilon,1+\epsilon] \cdot m \E[|p|^2].
\end{equation}
\end{theorem}
Proof of Theorem~\ref{thm:concentration_deg_d_poly_ell2} is deferred to Section~\ref{sec:concentration_low_deg}.
We next show a lower bound on  $\E[|p(x)|]$ and a concentration bound for $\sum_{i=1}^m |p(x_i)|$.
\begin{lemma}\label{lem:degree_d_ell1_bounds}
For any function $p$ of degree $d$, 
$$
\E_{x \sim \{0, 1\}^n}[|p(x)|] \in [ 3^{-d} ,1] \cdot \|\wh{p}\|_2.
$$
\end{lemma}
\begin{proof}
The upper bound comes by combining Jensen's inequality and Parseval's identity as follows: 
$$
\E_{x \sim \{0, 1\}^n}[|p(x)|] \le \E_{x \sim \{0, 1\}^n}[|p(x)|^2]^{1/2} = \|\wh{p}\|_2.
$$
To get the lower bound, we first recall Holder's inequality -- namely, if $r,s \ge 0$ and $1/r + 1/s = 1$, then 
\[
\E[|f(x)g(x)|] \le \E [|f(x)|^r]^{\frac{1}{r}} \E [|g(x)|^s]^{\frac{1}{s}}. 
\]
We apply Holder's inequality with $f(x)=|p(x)|^{2/3}, 1/r=2/3$ and $g(x)=|p(x)|^{4/3},1/s=1/3$ (with Parseval's identity) 
\begin{equation}~\label{eq:Holder}
\|\wh{p}\|_2^2=\E_{x \sim \{0, 1\}^n}[|p(x)|^2] \le \E_{x \sim \{0, 1\}^n}[|p(x)|]^{2/3}  \cdot \E_{x \sim \{0, 1\}^n}[|p(x)|^4]^{1/3}.
\end{equation}
Finally, we also recall the hypercontractivity theorem for low-degree polynomials~(see \cite{ODBook}). Namely, if $p: \mathbb{R}^n \rightarrow \mathbb{R}$ be a degree-$d$ polynomial. Then, 
\[
\E_{x \sim \{0, 1\}^n}[|p(x)|^4] \le 9^d \cdot \E_{x \sim \{0, 1\}^n}[|p(x)|^2]^2=9^d \|\wh{p}\|_2^4. 
\]
Combining this with (\ref{eq:Holder}), 
\[
\E_{x \sim \{0, 1\}^n}[|p(x)|] \ge \frac{\|\wh{p}\|_2^3}{\E_{x \sim \{0, 1\}^n}[|p(x)|^4]^{1/2}} \geq  3^{-d} \|\wh{p}\|_2. 
\]
\end{proof}

\begin{theorem}\label{thm:concentration_deg_d_poly_ell1}
There exists a constant $C$ such that for any $\epsilon, \delta>0$ and  $m = \frac{C {n \choose \le d} \cdot \log \frac{{n \choose \le d}}{\delta}}{\epsilon^2}$ random variables $x_1,\ldots,x_m \sim \{0, 1\}^n$, with probability $1-\delta$, for any degree $d$ function $p$,
\begin{equation}\label{eq:ell1norm_deg_d_poly}
\sum_{i=1}^m |p(x_i)| \in [1-\epsilon,1+\epsilon] \cdot m \cdot \E[|p(x)|].
\end{equation}
\end{theorem}
Proof of Theorem~\ref{thm:concentration_deg_d_poly_ell1} is deferred to Section~\ref{sec:concentration_low_deg}. We next use Theorems~\ref{thm:concentration_deg_d_poly_ell1} and ~\ref{thm:concentration_deg_d_poly_ell2} to argue that map defined by evaluation of low-degree polynomials at random points is a so-called \emph{Euclidean section}. More precisely, we show the following:
\begin{theorem}\label{thm:euclidean_section}
For any $d$ and $\delta$, there exists $m = O\big(\frac{{n \choose d} \cdot \log {n \choose d}}{\delta^2}\big)$ and $\rho=\frac{1}{4 \cdot 3^{2d}} - 8\delta$ such that for $m$ i.i.d. random points $x_1,\ldots,x_m \sim \{0,1\}^n$, 
with probability $0.99$, 
we have the following  guarantee: for any degree $d$ function $p$ and any subset $S \subseteq [m]$ of size $\rho m$, 
\begin{equation}~\label{eq:euclidean-1}
\sum_{i \in S} |p(x_i)| \le (1/2-\delta) \cdot \sum_{i \in [m]} |p(x_i)|.
\end{equation}
Further, 
\begin{equation}~\label{eq:euclidean-2}
\sum_{i \in S} |p(x_i)| + \delta \cdot m \E[|p(x) |] \le \sum_{i \in [m] \setminus S} |p(x_i)| \text{ and } \sum_{i \in S} |p(x_i)| + \delta \cdot 3^{-d} \cdot m \|\wh{p}\|_2 \le \sum_{i \in [m] \setminus S} |p(x_i)|.
\end{equation}
\end{theorem}
\begin{proof}
Sample $m$ independent random points $x_1, \ldots, x_m$ from $\{0,1\}^n$. Then, from Theorems~\ref{thm:concentration_deg_d_poly_ell2} and \ref{thm:concentration_deg_d_poly_ell1}, we get that with probability $0.99$, for any degree-$d$ function $p$, 
\begin{equation}\label{eq:property_deg_d}
\sum_{i=1}^m |p(x_i)| \in [1-\delta,1+\delta] \cdot m \E[|p(x)|] \text{ and } \sum_{i=1}^m |p(x_i)|^2 \in [1-\delta,1+\delta] \cdot m \|\wh{p}\|_2^2.
\end{equation}
 We are now ready to prove (\ref{eq:euclidean-1}). To do this, for contradiction, assume that there exists $p$ and $S$ of size $\rho m$ such that 
$$
\sum_{i \in S} |p(x_i)| \ge (1/2-\delta) \sum_{i=1}^m |p(x_i)|.
$$
Applying (\ref{eq:property_deg_d}) and Lemma~\ref{lem:degree_d_ell1_bounds} to the above, 
\begin{equation}\label{eq:adv_out_lower}
\sum_{i \in S} |p(x_i)| \ge (1/2-\delta) \cdot (1-\delta) m \E[|p(x)|] \ge (1/2-\delta) \cdot (1-\delta) m \cdot 3^{-d} \cdot \|\wh{p}\|_2.
\end{equation}
Next, by Cauchy-Schwartz inequality and (\ref{eq:property_deg_d}), we have 
\begin{equation}\label{eq:adv_out_upper}
(\sum_{i \in S} |p(x_i)|)^2 \le |S| \cdot \sum_{i \in S} |p(x_i)|^2 \le |S| \cdot \sum_{i \in [m]} |p(x_i)|^2 \le |S| \cdot (1+\delta) m \|\wh{p}\|_2^2.
\end{equation}
Combining (\ref{eq:adv_out_lower}) and (\ref{eq:adv_out_upper}), we have 
$$
|S| \ge  \frac{(1/2-\delta)^2 \cdot 3^{-2d} \cdot (1-\delta)^2 }{1+\delta} \cdot m > \frac{3^{-2d}}{4} - 8 \delta, 
$$
which contradicts the upper bound on the size of the set $S$. This finishes the proof of the first item. To prove the next two items, we have that 
$$
\sum_{i \in [m]\setminus S} |p(x_i)| - \sum_{i \in S}|p(x_i)|=\sum_{i \in [m]} |p(x_i)| - 2 \sum_{i \in S} |p(x_i)| \ge \delta \cdot m \cdot \E_x[|p(x)|]. $$
The inequality uses (\ref{eq:euclidean-1}). Using the lower bound of $\E_x[|p(x)|]$ from Lemma~\ref{lem:degree_d_ell1_bounds}, we get (\ref{eq:euclidean-2}). 

\end{proof}
We can now finish the proof of Theorem~\ref{thm:degree_d_adv_outliers}.

\begin{proofof}{Theorem~\ref{thm:degree_d_adv_outliers}}
Let $S$ denote the subset of adversarial outliers and $\vec{e}$ denote the noise vector on the observations $y(x_1),\ldots,y(x_m)$ such that $|\vec{e}(i)| \le \epsilon$ for any $i \notin S$. As $\rho=\frac{1}{4 \cdot 3^{2d}}-\delta$, Chernoff bound implies that  $|S| \le \frac{1}{4 \cdot 3^{-2d}}-\delta/2$ with probability $1-\exp(-\delta^2 m)$.
From Theorem~\ref{thm:euclidean_section} (instantiated with error $\delta/20$), with probability $0.99$, we have
\begin{equation}\label{eq:adv_gap_between_half_S}
\sum_{i \in S} |p(x_i)| + \frac{\delta}{20} \cdot m \E[|p(x)|] \le \sum_{i \in [m] \setminus S} |p(x_i)| \text{ and } \sum_{i \in S} |p(x_i)| + \frac{\delta}{20} \cdot 3^{-d} \cdot m \|\wh{p}\|_2 \le \sum_{i \in [m] \setminus S} |p(x_i)|.
\end{equation}
From the definition of $g$, we have 
\begin{equation}\label{eq:tilde_alpha_minimizer_GZ}
\sum_{i=1}^m |g(x_i)-y(x_i)| \le \sum_{i=1}^m |f(x_i)-y(x_i)|
\end{equation}
We set $h=f-g$ and lower bound the L.H.S. of \eqref{eq:tilde_alpha_minimizer_GZ} as
\begin{align*}
\sum_{i=1}^m \big|y(x_i)-g(x_i)\big| \ge \sum_{i \in S} \big|y(x_i) - f(x_i)\big| - \sum_{i  \in S} \big|h(x_i)\big| + \sum_{i \notin S} \big|h(x_i)\big| - \epsilon (m - |S|).
\end{align*}
Then we upper bound the R.H.S. of \eqref{eq:tilde_alpha_minimizer_GZ} as $ \sum_{i=1}^m |f(x_i)-y(x_i)| \le \sum_{i \in S} \big|y(x_i) - f(x_i)\big| + \epsilon (m - |S|)$. After plugging in these two bounds in \eqref{eq:tilde_alpha_minimizer_GZ}, we get
$$
\sum_{i \notin S} \big|h(x_i)\big| - \sum_{i  \in S} \big|h(x_i)\big| \le 2 \epsilon (m - |S|).
$$
Since~\eqref{eq:adv_gap_between_half_S} holds for any degree $d$ polynomial $p$, we plug in $p =h$ and get $\E[|h|] \lesssim \frac{\eps}{\delta}$ and $\|\wh{h}\|_2 \lesssim 3^{d} \cdot \frac{\eps}{\delta}$. The claim is now immediate. 
\end{proofof}

\subsection{Concentration of $\ell_1$ and $\ell_2$ Estimation for low-degree functions}\label{sec:concentration_low_deg}
We state the following version of the matrix Chernoff bound to prove Theorem~\ref{thm:concentration_deg_d_poly_ell2}.
\begin{theorem}[Theorem 1.1 of \cite{Tropp}]\label{thm:matrix_chernoff}
Consider a finite sequence $\{X_k\}$ of independent, random, self-adjoint matrices of dimension $d$. Assume that each random matrix $X_k$ satisfies (with probability $1$)
$$ X_k \succeq 0 \quad \text{ and } \quad \lambda(X_k) \le R.$$
Define $\mu_{\min}=\lambda_{\min}(\sum_k \E[X_k])$ and $\mu_{\max}=\lambda_{\max}(\sum_k \E[X_k])$. Then
\begin{align}
\Pr\left\{\lambda_{\min}(\sum_k X_k) \le (1-\eta)\mu_{\min}\right\} & \le d \left( \frac{e^{-\eta}}{(1-\delta)^{1-\eta}}\right)^{\mu_{\min}/R} \text{ for } \eta \in [0,1], and \\
\Pr\left\{\lambda_{\max}(\sum_k  X_k) \ge (1+\eta)\mu_{\max}\right\} & \le d \left( \frac{e^{-\eta}}{(1+\eta)^{1+\eta}}\right)^{\mu_{\max}/R} \text{ for } \eta \ge 0 
\end{align}
\end{theorem}
We will use the following notation -- given $x \in \{0,1\}^n$, let $\Mon_d(x) \in \mathbb{R}^{n \choose \le d}$ denote the  vector of all characters of degree at most $d$, i.e., $\big( (-1)^{\langle \xi , x \rangle} \big)_{|\xi| \le d}$. Observe that with this notation, for any degree $d$ function $p$ and any point $x$, $p(x)=\langle \wh{p}, \Mon_d(x) \rangle$.  We now give the proof of Theorem~\ref{thm:concentration_deg_d_poly_ell2}. 

\begin{proofof}{Theorem~\ref{thm:concentration_deg_d_poly_ell2}}
Begin by observing 
that $\underset{\E}{x}[|p(x)|^2]=\|\wh{p}\|_2^2$. Thus,  establishing \eqref{eq:ell2norm_deg_d_poly} for every degree-$d$ polynomial $p$ is equivalent to establishing that with probability $1-\delta$,  for every $\vec{\alpha} \in \mathbb{R}^{{n\choose \le d}}$, 
 $$
 \left\| \big( \Mon_d(x_1),\ldots,\Mon_d(x_m) \big)^{\top} \cdot \vec{\alpha} \right\|_2^2 \in [1-\epsilon,1+\epsilon] \cdot m \|\vec{\alpha} \|_2^2
$$ 
This in turn is equivalent to establishing tight upper and lower bounds on the eigenvalues of the matrix 
$\big( \Mon_d(x_1),\ldots,\Mon_d(x_m) \big) \cdot \big( \Mon_d(x_1),\ldots,\Mon_d(x_m) \big)^{\top}=\sum_{i=1}^m \Mon_d(x_i) \cdot \Mon_d(x_i)^{\top}$. Towards this, note that for all $i$,  $\|\Mon_d(x_i)\|_2^2={n \choose \le d}$ and further $\E_{x \in \{0,1\}^n}[ \Mon_d(x) \cdot \Mon_d(x)^\top ]=I_{{n \choose \le d} \times {n \choose \le d}}$. We can now apply the matrix Chernoff bound 
(Theorem~\ref{thm:matrix_chernoff}) to the matrix random variable  $\Mon_d(x)$  (where $x \sim \{0,1\}^n$)
 with error parameter $\eta = \epsilon$, confidence $\delta$ and $R = \binom{n}{ \le d}$. It suffices to set 
$m = C \frac{{n \choose \le d} \log \frac{{n \choose \le d}}{\delta}}{\epsilon^2}$ (for a sufficiently large constant $C$) to get the final result. 
\end{proofof}
To prove Theorem~\ref{thm:concentration_deg_d_poly_ell1}, we need the following result from \cite{CohenPeng15} (Theorem~1.1 in \cite{CohenPeng15}). 
\begin{theorem}\label{thm:ell_1_linear_family}
Given any $k$ characters $\chi_1,\ldots,\chi_k$, $\delta>0$, and $\eps>0$, for $m=O(\frac{k \log k/\delta}{\eps^2})$ randomly chosen points $x_1,\ldots,x_m$ in $\{0,1\}^n$, with probability $1-\delta$, 
$$ \forall f \in \textsf{span}\{\chi_1,\ldots,\chi_k\}, \sum_{i=1}^m |f(x_i)|\in [1- \eps, 1+\eps] m \cdot \E_x[|f(x)|].$$
\end{theorem}
We remark that in fact, \cite{CohenPeng15}  state a more general concentration result in terms of so-called ``Lewis weights". We get the above theorem by observing that the Lewis weight of $k$ characters $\chi_1,\ldots,\chi_k$ for the subspace $\textsf{span}\{\chi_1,\ldots,\chi_k\}$ is uniform over $\{0,1\}^n$. 
Instantiating Theorem~1.1 in \cite{CohenPeng15} in this setting gives us the above theorem. 

\begin{proofof}{Theorem~\ref{thm:concentration_deg_d_poly_ell1}} 
Let us define the set $\{\chi_1, \ldots, \chi_k\}$ as the set of all characters of Hamming weight at most $d$. Note that the size $k = {n \choose \le d}$. Observe that by definition, $p$ lies in the span of $\{\chi_1, \ldots, \chi
_k\}$. By now applying Theorem~\ref{thm:ell_1_linear_family} to this set of characters, Theorem~\ref{thm:concentration_deg_d_poly_ell1} follows. 
\end{proofof}

%% file: adv_outliers_torus.tex
\section{Sparse FFT of periodic signal under adversarial outliers}\label{sec:periodic_adv}
In this section, we consider recovery of Fourier sparse signals under adversarial outlier noise over the torus assuming an extra structural assumption: Namely, the Fourier coefficients are \emph{granular}, i.e., there is some number $\eta$ such that all the Fourier coefficients are integral multiples of this quantity. {Using complex analytic methods, we show that under such an assumption (Theorem~\ref{thm:periodic_adversarial_outlier}), we can tolerate a constant fraction of adversarial outliers. Further, Claim~\ref{clm:counter-ex}  shows that without this assumption, such an error tolerance cannot be achieved. In fact, note that even the  \emph{uncertainty principle} predicts that only about $1/k$ fraction of outliers can be tolerated and thus our theorem lets us \emph{beat the uncertainty principle} under the granularity assumption.  }

Let us consider $f: [0,1) \rightarrow \mathbb{C}$ where $f(t)=\sum_{j=1}^k \wh{f}(\xi_j) \cdot e^{2 \pi \bi \cdot \xi_j t}$. We assume that $f(t)$ satisfies the following two properties:
\begin{enumerate}
\item $\xi_1<\cdots<\xi_k$ are integer frequencies in the bandlimit $[-F,F]$.
\item Amplitudes $\wh{f}(\xi_1),\ldots,\wh{f}(\xi_k)$ are multiples of $\eta$ (in real and imaginary part separately) and $\sum_j |\wh{f}(\xi_j)|^2~\le~1$.
\end{enumerate}
\begin{theorem}\label{thm:periodic_adversarial_outlier}
There is an algorithm which given as input, sparsity parameter $k$, and additional parameters $\delta>0$, $\eta>0$, and $\rho < \frac{1}{2}-\delta$ has the following guarantee: There exist a positive constant $C_{\delta,\rho}$ and $\epsilon = \big(C_{\delta,\rho} \cdot \eta \big)^{O(1/\delta)}$ such that the algorithm can recover $f$ under any $(\rho,\epsilon)$-adversarial outlier noise if $f$ satisfies conditions (1) and (2) above.  The query complexity of the algorithm is 
$O(\frac{k \log F/\eta}{\delta^2})$. In fact, the queries to the oracle are just distributed uniformly in $[0,1)$.  
\ignore{
Given  and an algorithm with  samples such that for any $f$ whose amplitudes are multiples of $\eta$ and any $(\rho,\epsilon)$ adversarial outlier noise $e$, it outputs $f$ with probability $0.99$.}
\end{theorem}

\begin{algorithm}[H]
\caption{Sparse granular FFT  over $[0,1)$}\label{alg:BSS}
\begin{algorithmic}[1]
\Procedure{\textsc{SparseGranularFFT}}{$y,\delta,\eta$}
\State Sample $m=O(k \log F/\eta)$ random points $x_1,\ldots,x_m$ from $[0,1)$;
\State  Let $y(\cdot)$ be the oracle and $y_1=y(x_1), \ldots, y_m=y(x_m)$. 
\For {$g$ with $\xi_1,\ldots,\xi_k \in [-F,F]$ and amplitudes of multiples of $\eta$ satisfying $\underset{j}{\sum} |\wh{g}(\xi_j)|^2 \le 1$}
\State $s = |\{i : |g(x_i) - y_i| \ge \eps\}$;
\If {$s<(0.5-\delta/2)m$}
\State Output $g$;
\EndIf
\EndFor
\EndProcedure
\end{algorithmic}
\end{algorithm}
Note that the algorithm above is computationally inefficient as it iterates over $\approx (|F|/\eta)^k$ functions in order to find the best $g$. 
The main technical ingredient used to prove Theorem~\ref{thm:periodic_adversarial_outlier} is the following anti-concentration lemma for periodic signals with integral amplitudes.
\begin{lemma}\label{lem:anticoncentration_period_signals}
Let $f(t)=\sum_{j=1}^k v_j e^{2 \pi \bi \cdot \xi_j t}$ with integer frequencies $\xi_1<\cdots<\xi_k$ and amplitudes $\sum_j |v_j|^2=1$. For any constant $\alpha \in (0,1)$, there exists $C_\alpha$ such that for any $\eta$, if $|v_1| \ge \eta$, 
then for $\delta=(\eta/C_\alpha)^{1/\alpha}$, 
$$
\underset{t \sim [0,1]}{\Pr} \big[ |f(t)| \le \delta \big] \le \alpha.
$$
\end{lemma}

\begin{proof}
Let $P(z)=\sum_{j=1}^k v_j z^{\xi_j-\xi_1}$ and note  that $P(e^{2 \pi \bi t})=e^{2 \pi \bi \xi_1 t} f(t)$. Let $\alpha$ denote the fraction of $z$ on the unit disc with $|P(z)| \le \delta$.  Observe that $P(z)$ is a polynomial in $z$ and thus, we can apply Jensen's formula~\cite{ahlfors1979complex} to get
\begin{equation}\label{eq:subharmonic}
\log |P(0)| \le \frac{1}{2 \pi} \int_{0}^{2\pi} \log |P(e^{\bi \theta})| \mathrm{d} \theta.
\end{equation}
Now, observe that the left hand side is at least $\log \eta$. We will now upper bound the right hand side in terms of $\alpha$. To do this, observe by Parseval's theorem,  we have $\underset{z}{\E}[|P(z)|^2 ]=1$. This implies that 
\[
\underset{z}{\E}\big[|P(z)|^2 \big | |P(z)| \ge \delta \big] \le \frac{1}{1-\alpha}. 
\]
Using the fact that $\log (\cdot)$ is a concave function, we have by Jensen's inequality that 
\begin{eqnarray*}
\underset{z}{\E}\big[\log |P(z)| \ \big  | \ |P(z)| \ge \delta \big]  &=& \frac12 \cdot \underset{z}{\E}\big[\log |P(z)|^2\ \big | \ |P(z)| \ge \delta \big] \\ &\le&  \frac12 \cdot \log \big( \underset{z}{\E}\big[ |P(z)|^2 \ \big |\ |P(z)| \ge \delta \big]\big)  \\ 
&=& \frac12 \cdot \log \frac{1}{1-\alpha}. 
\end{eqnarray*} 
Since $\log |P(z)| \le \log \delta$ whenever $|P(z) | \le \delta$, we get that 
\[
\underset{z}{\E}\big[\log |P(z)| \big]  \leq (1-\alpha) \cdot  \frac12 \cdot \log \frac{1}{1-\alpha} + \alpha \log \delta. 
\]
Applying (\ref{eq:subharmonic}), we have 
$$
\log \eta \le \alpha \cdot \log \delta + (1-\alpha) \cdot \frac{1}{2} \log \frac{1}{1-\alpha}.
$$
Set $C_\alpha = (\frac{1}{1-\alpha})^{(1-\alpha)/2}$ and we get that $\delta \ge (\eta/C_{\alpha})^{1/\alpha}$ which finishes the proof. 
\ignore{

 Now observe that 
for $z$ on the unit disc with $|P(z)| \le \delta$, $\log |P(z)| \le \log \delta$.

Next we consider $z$ on the unit disc with $|P(z)| \ge \delta$. At first, we have $\E_z \big[ |P(z)|^2 \big] = 1$ from , which indicates $\underset{z}{\E}\bigg[|P(z)|^2 \bigg | |P(z)| \ge \delta \bigg] \le \frac{1}{1-\alpha}$ given $\alpha$ is the  fraction of $z$ with $|P(z)| \le \delta$. Since $\log |P(z)|=\frac{1}{2} \cdot \log |P(z)|^2$ and $\log$ is a concave function, $\underset{z}{\E}\bigg[\log_2 |P(z)| \bigg | |P(z)| \ge \delta \bigg] \le \frac{1}{2} \log \frac{1}{1-\alpha}$.

From all discussion above, we rewrite \eqref{eq:subharmonic} using $|P(0)| \ge \eta$:

We simplify it to $\log \eta/C_\alpha \le \alpha \cdot \log \delta$ for $C_{\alpha}=(\frac{1}{1-\alpha})^{(1-\alpha)/2}$, which indicates }
\end{proof}
We now use Lemma~\ref{lem:anticoncentration_period_signals} to finish the proof of Theorem~\ref{thm:periodic_adversarial_outlier}. 

\begin{proofof}{Theorem~\ref{thm:periodic_adversarial_outlier}}
As the algorithm iterates over $g$ in Step~4 of the algortihm, it is clear that iterates over $f$.  Now, consider any other $g \not =f$ (in this enumeration). Note that $g-f$ satisfies two properties: (i) 
$\|g-f\|_2 \le \|f\|_2+\|g\|_2 \le 2$; (ii) The first non-zero coefficient of $g-f$ is at least $\eta$ in magnitude. Applying Lemma~\ref{lem:anticoncentration_period_signals} to $\frac{g-f}{\|g-f\|_2}$ (whose first non-zero coefficient is at least $\eta/2$) with $\alpha=\delta/4$, we conclude that 
$
\Pr_{t \in [0,1)} [|g(t)-f(t)| \le 2\eps] \le  \delta/4, 
$  for our choice of $\eps$. 

Now, consider the case when $g = f$. Then, note that the set $s$ defined in Step~5 of the algorithm contains only points corrupted by the adversarial outliers, which has size at most $(0.5 - \delta/2)m$ with probability $1- \exp(-\delta^2m)$. On the other hand, consider any $g \not =f $. Then, with probability $1- \exp(-\delta^2m)$, for at least $(1-\delta/2)m$  of the points $x_1, \ldots, x_m$, $|f(x_i) - g(x_i)| > 2\eps$. Also, with probability $1- \exp(-\delta^2m)$, for at most $(1/2-\delta/2)m$ of the points corrupted by adversarial outliers, $|f(x_i) - y(x_i)| > \eps$. Thus, for at least $m/2$ of the points  $x_1, \ldots, x_m$, $|y(x_i) - g(x_i)| > \eps$.  This means that such a $g$ will not be output by the algorithm, except with probability $1-\exp(\delta^2 m)$.  

 Finally, we observe that the set of all functions over  which the algorithm enumerates 
is at most $\binom{2F}{k} \cdot (1/\eta)^k$. Taking a  union bound over all functions in this set, we get that the algorithm outputs the correct $f$ with probability at least $0.99$. 
\ignore{
As the algorithm iterates over $g$ in Step~4 of the algortihm, it is clear that iterates over $f$.  Now, consider any other $g \not =f$ (in this enumeration). Note that $g-f$ satisfies two properties: (i) 
$\|g-f\|_2 \le \|f\|_2+\|g\|_2 \le 2$; (ii) The first non-zero coefficient of $g-f$ is at least $\eta$ in magnitude. Applying Lemma~\ref{lem:anticoncentration_period_signals} to $\frac{g-f}{\|g-f\|_2}$ (whose first non-zero coefficient is at least $\eta/2$) with $\alpha=1/2-\rho-\delta$, we conclude that 
$
\Pr_{t \in [0,1)} [|g(t)-f(t)| \le 2\epsilon] \le 1/2 -\rho -\delta, 
$  for our choice of $\eps$. 

This means that with probability $1-\exp(-\delta^2 \cdot m)$, at least $(1/2- \delta/2)m$ points 


With probability $\exp(-\delta^2 \cdot m)$, at least $(1/2-\delta/2)m$ points have observations $\epsilon$-far away from $g(t)$. From a union bound over all possible signals $g$, we prove the correctness of our algorithm.}
\end{proofof}

\subsection{Necessity of a lower bound on the amplitudes of non-zero frequencies}
We now provide an example to show the lower bound on the amplitude of the first non-zero frequency is necessary in order to tolerate a constant fraction of outliers. Note that our algorithmic upper bound requires that non-zero amplitudes be integral multiple of some fixed $\eta$ (as opposed to just being larger than $\eta$). This is because we apply Lemma~\ref{lem:anticoncentration_period_signals} not just to the target function $f$ but rather $f-g$ where $g$ is some other potential function. If all non-zero amplitudes are integral multiples of $\eta$, then for any $f \not = g$, the first non-zero amplitude of $f-g$ is necessarily larger than $\eta$. 
\begin{claim}~\label{clm:counter-ex}
For any $\alpha<1$, there exist a constant $C'_{\alpha}$ and a $k$-Fourier-sparse function $f(t)=\sum_{j=1}^k v_j e^{2 \pi \bi \cdot \xi_j t}$ with integer frequencies $\xi_1<\cdots<\xi_k$ and amplitudes $\sum_j |v_j|^2=1$ such that 
$$
\Pr_{t \sim [0,1)} \big[ |f(t)| \le 2^{-C'_{\alpha} \cdot k} \big] \ge \alpha.
$$
\end{claim}
\begin{proof}
Let us consider $g(t)=(1+e^{2\pi \bi t})^k=\sum_{j=0}^k {k \choose j} e^{2 \pi \bi j t}$. Clearly, $g(\cdot)$ is $k$-Fourier sparse. Next, for any $t \in [0,1]$, we have
\begin{align*}
|g(t)| & =|1+e^{2\pi \bi t}|^k\\
& = |1+\cos(2 \pi t) - i \sin(2\pi t)|^k\\
& = \left( 1 + 2 \cos(2 \pi t) + \cos^2(2 \pi t) + \sin^2(2 \pi t) \right)^{k/2}\\
& = \left(2 + 2 \cos(2 \pi t) \right)^{k/2}.
\end{align*}
Now,   observe that $\sum_{\xi} |\widehat{g}(\xi)|^2 = \sum_{j=0}^k {k \choose j}
^2 = \sum_{j=0}^k {k \choose j} {k \choose k-j}={2k \choose k} \ge\frac{1}{2\sqrt{k}} 2^{2k}$. Define $f(t) = \frac{g(t)}{\sqrt{\sum_{j=0}^k {k \choose j}
^2}}$. Clearly $\widehat{f}$ is $k$-sparse and $\sum_{\xi} |\widehat{f}(\xi)|^2=1$. Finally,  we have 
\[
|f(t)| \le \sqrt{2} \cdot  k^{\frac14}  \cdot  \frac{|g(t)|}{2^k} =  \sqrt{2} \cdot  k^{\frac14}  \cdot\left( \frac{1 + \cos(2 \pi t)}{2} \right)^{k/2} =\sqrt{2} \cdot  k^{\frac14}  \cdot \sin^{k}(\pi t).  
\]
Observe that the event $\mathcal{E} = t \in [\frac{1-\alpha}{2},\frac{1+\alpha}{2}]$ happens with probability $\alpha$. Further, $|f(t)|$, conditioned on $\mathcal{E}$ is at most $2^{- C'_\alpha \cdot k}$, where $C'_{\alpha}= (1-\Theta(\alpha)^2)/2$. Since $\Pr[\mathcal{E}] =\alpha$, this finishes the proof. 
\end{proof}


%% file: appen_proof.tex
\section{Symmetrization and Gaussianization}\label{appd:sym_gau}
In this section, we prove Theorem~\ref{thm:sym_gauss}. 
Let us begin by recalling the theorem statement. 
\restate{thm:sym_gauss}

\begin{proof}
We first use the convexity of the $|\cdot |$ function to move out $\underset{x'}{E}$:
\begin{align*}
\underset{x}{\E}\left[\max_{\Lambda} \left|\sum_{j=1}^n f(\Lambda,x_j) - \underset{x'}{\E}[\sum_{j=1}^n f(\Lambda,x'_j)] \right|\right] & \le \underset{x}{\E}\left[\max_{\Lambda} \underset{x'}{\E} \left|\sum_{j=1}^n f(\Lambda,x_j) - \sum_{j=1}^n f(\Lambda,x'_j) \right|\right] \ \ \textrm{using convexity of} \ |\cdot|\\
& \le \underset{x,x'}{\E}\left[\max_{\Lambda} \left|\sum_{j=1}^n f(\Lambda,x_j) - \sum_{j=1}^n f(\Lambda,x'_j) \right|\right] \ \ \textrm{using concavity of} \ \max(\cdot)\\
& \le \sqrt{\frac{\pi}{2}} \underset{x,x'}{\E}\left[\max_{\Lambda} \left|\sum_{j=1}^n \big(f(\Lambda,x_j)  - f(\Lambda,x'_j) \big) \cdot\underset{g_j}{\E}|g_j| \right|\right] \ \ \textrm{uses} \E[|g_j|]=\sqrt{\frac{2}{\pi}} \\
& \le \sqrt{\frac{\pi}{2}} \underset{x,x'}{\E}\left[\max_{\Lambda} \underset{g}{\E} \left|\sum_{j=1}^n \big(f(\Lambda,x_j)  - f(\Lambda,x'_j) \big)\cdot |g_j| \right|\right] \\ & \ \ \textrm{using convexity of} \ |\cdot|\\
& \le \sqrt{\frac{\pi}{2}} \underset{x,x'}{\E} \underset{g}{\E} \left[\max_{\Lambda} \left|\sum_{j=1}^n \big(f(\Lambda,x_j)  - f(\Lambda,x'_j) \big)\cdot |g_j| \right|\right]\\
& \ \ \textrm{using concavity of} \ \max(\cdot)\\
& = \sqrt{\frac{\pi}{2}} \underset{g}{\E} \underset{x,x'}{\E}  \left[\max_{\Lambda} \left|\sum_{j=1}^n \big(f(\Lambda,x_j) -  f(\Lambda,x'_j) \big) \cdot g_j\right|\right]\\ & \textrm{by symmetry of } f(\Lambda,x_j)  - f(\Lambda,x'_j) \\ 
& \le \sqrt{\frac{\pi}{2}}\underset{x,x'}{\E}  \underset{g}{\E} \left[\max_{\Lambda} \left|\sum_{j=1}^n f(\Lambda,x_j) g_j\right| + \max_{\Lambda} \left|- \sum_{j=1}^n f(\Lambda,x'_j) g_j \right|\right] \\ & \  \textrm{using triangle inequality}\\
& \le \sqrt{2 \pi} \underset{x}{\E}  \underset{g}{\E} \left[\max_{\Lambda} \left|\sum_{j=1}^n f(\Lambda,x_j) g_j\right| \right].
\end{align*}
\end{proof}